\documentclass[a4paper,superscriptaddress,amsmath,amssymb,aps, twocolumn]{revtex4-2}
\usepackage{amsfonts,amsmath,amssymb,amsthm}
\usepackage{braket}
\usepackage{physics}
\usepackage{bm}
\usepackage{graphicx}
\usepackage{color}
\usepackage[colorlinks=true]{hyperref}
\usepackage{algorithm}
\usepackage{algpseudocodex}
\usepackage{quantikz}
\usepackage{enumitem}
\usepackage{booktabs, multirow}
\usepackage{natbib}

\hypersetup{
           breaklinks=true,   
           colorlinks=true,   
           pdfusetitle=true, 
           citecolor=blue,
           urlcolor=blue
           }

\theoremstyle{plain}
\newtheorem{thm}{Theorem}

\newtheorem{lemma}[thm]{Lemma}

\theoremstyle{definition}

\theoremstyle{remark}

\algdef{SE}[DOWHILE]{Do}{doWhile}{\algorithmicdo}[1]{\algorithmicwhile\ #1}

\usepackage[normalem]{ulem}

\begin{document}

\title{
Adaptive identification of low-degree polynomials in quantum singular value transformation: application to nonlinear quantum properties estimation
}

\author{Jumpei Kato}
\thanks{jumpei\_kato[at]keio.jp}
\affiliation{Mitsubishi UFJ Financial Group, Inc.~and MUFG Bank, Ltd., 4-10-2 Nakano, Nakano-ku, Tokyo 164-0001, Japan}
\affiliation{Quantum Computing Center, Keio University, Hiyoshi 3-14-1, Kohoku, Yokohama, Kanagawa 223-8522, Japan}
\affiliation{Graduate School of Science and Technology, Keio University, 3-14-1 Hiyoshi, Kohoku, Yokohama, Kanagawa 223-8522, Japan}

\author{Akira Tanji}
\affiliation{Graduate School of Science and Technology, Keio University, 3-14-1 Hiyoshi, Kohoku, Yokohama, Kanagawa 223-8522, Japan}

\author{Hiroyuki Harada}
\affiliation{Graduate School of Science and Technology, Keio University, 3-14-1 Hiyoshi, Kohoku, Yokohama, Kanagawa 223-8522, Japan}

\author{Kaito Wada}
\affiliation{International Center for Elementary Particle Physics, University of Tokyo, 7-3-1 Hongo, Bunkyo-ku, Tokyo 113-0033, Japan}
\affiliation{Graduate School of Science and Technology, Keio University, 3-14-1 Hiyoshi, Kohoku, Yokohama, Kanagawa 223-8522, Japan}

\author{Kosuke Ito}
\affiliation{Advanced Material Engineering Division, Toyota Motor Corporation, 1200 Mishuku, Susono, Shizuoka 410-1193, Japan}
\affiliation{Quantum Computing Center, Keio University, Hiyoshi 3-14-1, Kohoku, Yokohama, Kanagawa 223-8522, Japan}

\author{Naoki Yamamoto}
\thanks{yamamoto[at]appi.keio.ac.jp}
\affiliation{Quantum Computing Center, Keio University, Hiyoshi 3-14-1, Kohoku, Yokohama, Kanagawa 223-8522, Japan}
\affiliation{Department of Applied Physics and Physico-Informatics, Keio University, Hiyoshi 3-14-1, Kohoku, Yokohama, Kanagawa 223-8522, Japan}
\date{\today}
\begin{abstract}
    Estimating properties of unknown quantum states via quantum singular value transformation (QSVT) often requires high-degree polynomials to handle small eigenvalues of density matrices.
    Specifically, the existing approaches determine the polynomial degree by relying on overly conservative worst-case bounds based on the minimum non-zero eigenvalue or the rank of the density matrices.
    In this work, we propose a spectral cutoff method that truncates the negligible eigenvalue tail depending on the task, the target accuracy, and the state, which enables the use of significantly lower-degree polynomials. 
    To implement this, we develop a two-stage algorithm to estimate nonlinear properties, particularly von Neumann entropy and R{\'e}nyi entropy.
    In the first stage, we execute a search algorithm to identify the spectral cutoff directly from the unknown quantum state.
    In the second stage, we estimate the nonlinear properties utilizing QSVT with the degree of the polynomial adaptively determined by the cutoff.
    This two-stage algorithm enables adaptive property estimation without prior knowledge of the minimum eigenvalue or the rank, using low-degree polynomials beyond conventional worst-case choices.
\end{abstract}
\maketitle

\section{Introduction}

\begin{figure}[t]
    \centering
    \includegraphics[width=1\columnwidth]{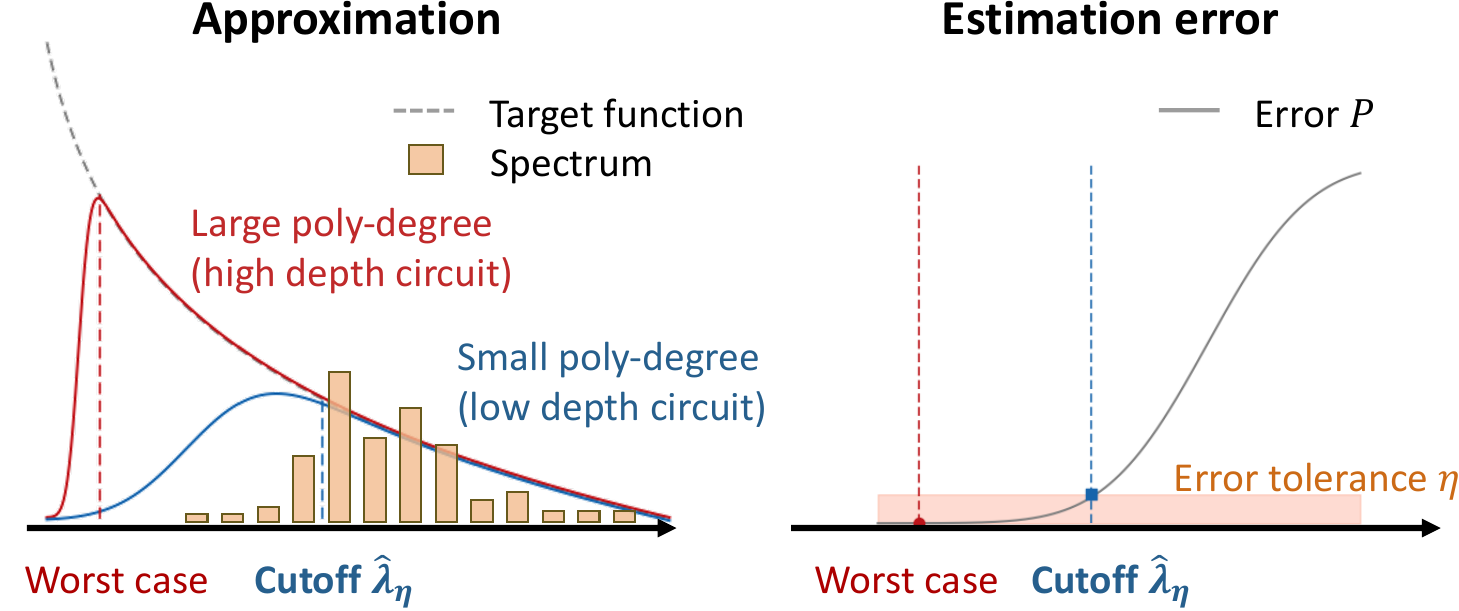}
    \caption{The illustration of the idea. 
    (Left) When using a polynomial to approximate a target function that is singular at the origin, the conventional approach truncates the polynomial degree based on a worst-case threshold below the minimum eigenvalue of the state.
    (Right) Our proposed algorithm adaptively identifies the cutoff $\hat{\lambda}_\eta$ with reference to the estimation error tolerance $\eta$ of the target nonlinear property, while the conventional one takes a property-independent worst-case threshold. 
    The cutoff $\hat{\lambda}_\eta$ yields a sufficiently low-degree polynomial in the left figure.}
    \label{fig:idea}
\end{figure}

\subsection{Background}\label{sec:background}
Extracting meaningful features from unknown quantum states is of primary interest in quantum science and technology~\cite{Bravyi2011-testing, montanaro2013-survey,anshu2024survey}.
In particular, nonlinear properties of quantum states, such as entropies and distance metrics, provide rich information about the underlying quantum systems.
Establishing an efficient method to estimate these properties directly from multiple identical copies of unknown states is a central challenge in the field.

The baseline approach to this problem has traditionally relied on quantum state tomography~\cite{haah2016sample,o2016efficient,pelecanos2025mixed} to fully reconstruct the underlying density matrix and then classically evaluate the target properties.
However, the tomographic approach often requires exponentially many state copies in the system size.
To avoid this prohibitive cost, significant progress has been made in developing quantum property estimation algorithms, such as shadow tomography~\cite{10.1145/3188745.3188802,aaronson2019gentle,huang2020predicting}, an empirical Young diagram approach~\cite{acharya2020-entropy},
variational approaches~\cite{Cerezo2020-variational, Tan2021-variational, goldfeld2024neural}, and multi-observable estimation under coherent access~\cite{PhysRevLett.129.240501,PRXQuantum.6.020308,4c6g-zx6c}.

A representative approach for estimating nonlinear state properties is to employ the quantum singular value transformation (QSVT)~\cite{gilyen2019-qsvt, low2019qubitization}.
The QSVT serves as an algorithmic framework to apply a wide range of polynomial transformations to the singular values (or eigenvalues for Hermitian matrices) of an input matrix. 
Many QSVT-based algorithms have been developed to estimate properties of quantum states such as entropy~\cite{gilyen2019-world, gur2021-sublinear, Li2019-query, Subramanian2021-renyi, wang2024-new, wang2024discrete, nghiem2025measuringancilla, wang2025entropy-samplizer}, trace distance and fidelity~\cite{wang2023fasttrace, gilyen2022-fidelity, Wang2023-fidelity, utsumi2025uulmann}, and more~\cite{gilyen2022petz, Liu2024-matrix-geometric-means}.
Some algorithms exhibit efficient complexity, avoiding the exponential scaling with the system size.

These QSVT-based algorithms are designed by approximating target functions with polynomials.
The degree of the polynomial approximation is a key parameter that determines the quality of the approximation and the computational cost. 
Since most target functions of interest such as $\log (x)$ and $x^{\alpha-1}$ are singular or non-analytic at $x=0$, approximating them over a region containing very small eigenvalues requires a high polynomial degree. 
Accordingly, the computational cost of those QSVT-based algorithms depends on the tail contribution of the spectrum.
However, for an unknown quantum state, such spectral information is generally not available beforehand. 
Existing analyses therefore often rely on coarse prior information, such as the minimum non-zero eigenvalue $\lambda_{\min}$ or the rank $r$ of a density matrix. 
While the polynomial degree based on such parameters safely controls the approximation error, it may be overly conservative; see Fig.~\ref{fig:idea}~(left).

In practical problems, however, it is almost hopeless to obtain a good knowledge of $\lambda_{\min}$ or $r$ in advance. 
Without such prior knowledge, we need to take a very conservative lower bound of $\lambda_{\min}$ near $x=0$, which could greatly increase the degree. 
We may also assume a full-rank state, resulting in the system-dimensional (i.e., exponential) cost in QSVT algorithms similar to the tomographic approach. 
In addition, the parameters $\lambda_{\min}$ and $r$ are highly sensitive to noise; even small-intensity noise instantly converts a low-rank state into a full-rank state, driving the required complexity up to the system dimension.
Therefore, a smoothed input parameter is highly desirable.

\subsection{Summary of the work}
This work provides a framework to identify low-degree polynomials sufficient for estimating nonlinear properties of a given unknown state $\rho$ and its copies by finding an appropriate spectral cutoff. 
A rough overview is illustrated in Fig.~\ref{fig:idea}~(right). 
The cutoff $\hat{\lambda}_\eta$ truncates the negligible eigenvalues of $\rho$ within an estimation error tolerance $\eta$. 
Since this cutoff is identified by the spectrum of $\rho$ rather than the worst-case parameters characterized by $\lambda_{\min}$ or $r$,
the polynomial degree determined by $\hat{\lambda}_\eta$ can be significantly lower than the worst-case bound based on $\lambda_{\min}$ or $r$.
Furthermore, the cutoff is expected to be highly robust to noise. 
In particular, for a noisy low-rank state, the algorithm can avoid the above-mentioned system-dimensionality cost caused by the noise.

To implement this idea in practice, we propose an adaptive two-stage algorithm to efficiently estimate properties of quantum states. 
In Stage~1 (identification), our search algorithm identifies the spectral cutoff $\hat{\lambda}_\eta$ in an iterative way using the spectrum of $\rho$. 
In Stage~2 (estimation), the algorithm executes the property estimation via QSVT using a polynomial degree determined by the identified cutoff.

Note that the two-stage strategy is a common design principle in quantum estimation and related algorithms. 
For example, in quantum parameter estimation~\cite{Barndorff-2000}, a preliminary stage can be used to obtain a coarse estimate of an unknown parameter, which is then used to design a near-optimal measurement in the second stage.

We apply the two-stage algorithm to estimate von Neumann entropy and R{\'e}nyi entropy. 
Focusing on this specific target property reveals the necessary number of unknown copies, governed by the identified spectral cutoff, which depends on the target task, quantum state, and target accuracy.
Since the cutoff can be much larger than the worst-case bound, our algorithm reduces the polynomial degree and thereby the copy requirements. 

\subsection{Relation to prior works}

Our main application is estimating von Neumann entropy in the sample-access model, where the input is given as independent copies of an unknown quantum state $\rho$.
Several recent works have studied this problem.
While some algorithms pay the system-dimensional cost~\cite{acharya2020-entropy, wang2025entropy-samplizer}, other works assume the prior knowledge of rank $r$ or the minimum non-zero eigenvalue $\lambda_{\min}$ to avoid it.
Ref.~\cite{wang2023entropy} proposes Fourier-based algorithms for estimating von Neumann entropy with accuracy $\varepsilon$ using $\tilde{\mathcal{O}}(\lambda_
{\min}^{-2}\varepsilon^{-5})$ and $\tilde{\mathcal{O}}(r^2\varepsilon^{-7})$ copies of $\rho$. 
These copy requirements have been improved to $\tilde{\mathcal{O}}(\lambda_{\min}^{-2}\varepsilon^{-4})$~\cite{nghiem2025measuringancilla} 
and $\tilde{\mathcal{O}}(r^2\varepsilon^{-5})$~\cite{wang2025entropy-samplizer}, respectively.
Recently, Ref.~\cite{wada2025state} achieves the copy requirements of $\tilde{\mathcal{O}}(\lambda_{\min}^{-2}\varepsilon^{-2})$ via virtual density matrix exponentiation technique, and Ref.~\cite{Kean2025-list} derives an upper bound of $\tilde{\mathcal{O}}(r^{2}\varepsilon^{-4})$ via the sample-to-query lifting.

Our algorithm achieves the cutoff-dependent scaling $\tilde{\mathcal{O}}(\hat{\lambda}_{\eta}^{-2}\varepsilon^{-2})$
using an identified spectral cutoff $\hat{\lambda}_\eta$; see Lemma~\ref{lemma:vn}. 
We also derive lower bounds on $\hat{\lambda}_{\eta}$ given in Theorem~\ref{thm:search}, showing that our copy complexity recovers the recent scalings $\tilde{\mathcal{O}}(\lambda_{\min}^{-2}\varepsilon^{-2})$ or $\tilde{\mathcal{O}}(r^{2}\varepsilon^{-4})$ derived in \cite{wada2025state, Kean2025-list} in the worst-case.
Because the identified cutoff can be substantially larger than the worst-case bound depending on states,
our algorithm has a potential to largely improve the total copy requirement. 

\section{Preliminary}\label{sec:preliminary}
\subsection{Problem setup}
Let $\rho$ be an $N$-dimensional density matrix acting on an $n$-qubit system, where $N = 2^n$.
In this work, we primarily focus on the sample-access model, where we are given identical copies of an unknown quantum state $\rho$.
The results in this paper can be directly extended to the coherent-access model, where one has direct access to a block-encoding~\cite{gilyen2019-qsvt} of $\rho$.

Our goal is to estimate nonlinear properties of $\rho$, which are written as $\tr[f(\rho) \rho ]$ with a target function $f$. 
For example, the target function $f$ for the von Neumann entropy is given by $f(x)=-\log x$, and $f(x)=x^{\alpha-1}$ for the R{\'e}nyi moment of order $\alpha$.

\subsection{QSVT-based property estimation}

To estimate $\tr[f(\rho) \rho ]$, we employ the QSVT-based algorithms. 
Specifically, we use a polynomial $F(\rho)$ approximating $f(\rho)$,
followed by estimating $\tr[F(\rho) \rho]$ that approximates the target quantity $\tr[f(\rho) \rho]$.
This framework is widely applicable to estimating nonlinear quantum properties~\cite{wang2024-new}.

In the sample-access model, the density matrix exponentiation (DME) technique \cite{lloyd2014-qpca} plays a key role in translating the sample access to the block-encoding, which is a general input for QSVT~\cite{gilyen2022petz, gilyen2022-fidelity}. 
The DME is a protocol for constructing the time evolution generated by an unknown quantum state (density matrix) $\rho$, by consuming copies of $\rho$.
Ref.~\cite{lloyd2014-qpca} proposes an algorithm that constructs gate operations to approximate the unitary channel of $e^{-i\rho T}$ within $\varepsilon_{\diamond}$-diamond norm error, using $\Theta(T^2/\varepsilon_{\diamond})$ copies of $\rho$~\cite{Kimmel2017-hv}.
Recently, by focusing on the property estimation, Ref.~\cite{wada2025state} developed the virtual DME protocol that requires only $\mathcal{O}(T^2\log(T/\varepsilon_{\diamond})) = 
\tilde{\mathcal{O}}(T^2)$ copies of $\rho$ to effectively simulate $e^{-i\rho T}$.
We use the tilde notation $\tilde{\mathcal{O}}(\cdot)$ that hides the logarithmic factors
in the problem parameters other than the failure probability, i.e., the dependence on $\vartheta$ is displayed explicitly.
In exchange for the logarithmic dependence on $\varepsilon_{\diamond}$, the virtual DME is limited to expectation-value estimation tasks, and it must be combined with the subsequent measurement process. 
However, for our purpose of property estimation, the virtual DME is sufficient. 

In practice, QSVT can implement only polynomials satisfying the standard admissibility conditions~\cite{gilyen2019-qsvt}. 
Moreover, in the sample-access model, the block-encoding of $\rho$ obtained from DME naturally comes with a constant normalization factor $\pi$~\cite[Corollary 71]{gilyen2019-qsvt}. 
We therefore choose an implementable polynomial $G$ that approximates a rescaled and shifted version of the target function:
\begin{equation}\label{eq:f-to-G}
    f(x) \approx   F(x) = \alpha_F G(x/\pi) + \beta_F,
\end{equation}
where $\alpha_F$ and $\beta_F$ are constants. 
The following lemma gives the copy complexity for estimating a nonlinear property 
$\tr[(\alpha_F G(\rho/\pi) + \beta_F)\rho] \approx \tr[f(\rho)\rho]$ with $G$ a QSVT-implementable polynomial. 

\begin{lemma}[QSVT-based estimation via virtual DME]\label{lemma:virtual-qsvt}
    Let $\rho$ be an $N$-dimensional quantum state.
    Let $G(x)$ be a real polynomial of degree at most $D$ satisfying $|G(x)|\leq 1$ for all $x\in[-1,1]$ and the standard parity condition for QSVT.
    For $\varepsilon >0$ and $\vartheta \in (0,1)$, $\tr[(\alpha_F G(\rho/\pi) + \beta_F) \rho]$ can be estimated within additive error $\varepsilon$ with success probability at least $1 - \vartheta$.
    This uses 
    $\tilde{\mathcal{O}}\!\left( D^2\right)$ copies of $\rho$ per circuit and $\tilde{\mathcal{O}}\!\left(\alpha_F^2\varepsilon^{-2}\log(\vartheta^{-1})\right)$ measurement shots.
    Consequently, the total copy complexity is $\tilde{\mathcal{O}}\!\left(D^2\alpha_F^2\varepsilon^{-2}\log(\vartheta^{-1})\right)$.
\end{lemma}
\begin{proof}
    From Ref.~\cite[Corollary 71]{gilyen2019-qsvt}, given access to the controlled time-evolution $e^{i\rho/2}$, we can construct an $\varepsilon_{\rm BE}$-precise block-encoding of $\rho/\pi$,
    requiring $\mathcal{O}(\log(1/\varepsilon_{\rm BE}))$ uses of the controlled-$e^{i\rho/2}$ and its inverse.
    
    Since $G$ is a degree-$D$ polynomial, the block-encoding of $G(\rho/\pi)$ requires $\mathcal{O}(D)$ queries to the block-encoding of $\rho/\pi$ and its inverse, leading to $M = \mathcal{O}(D\log(1/\varepsilon_{\mathrm{BE}}))$ uses of $e^{i\rho/2}$. 
    We regard the QSVT circuit for $G(\rho/\pi)$ as an interleaving sequence of time evolutions generated by $\rho$ and unitaries independent of $\rho$.
    Ref.~\cite{wada2025state} shows that $M$-interleaving sequences of the time evolutions and the arbitrary unitaries can be simulated via virtual DME with $\varepsilon_{\diamond} > 0$ in the diamond norm,
    using $\mathcal{O}(M^2\log(M/\varepsilon_{\diamond}))$ copies per circuit.
    By setting $\varepsilon_{\mathrm{BE}} = \varepsilon/(4D \alpha_F)$ and  $\varepsilon_{\diamond} = \varepsilon/(4\alpha_F)$, we obtain the circuit with $\varepsilon/(2\alpha_F)$ approximation error in the diamond norm, which gives  $\tilde{\mathcal{O}}(D^2(\log(\alpha_FD/\varepsilon))^3)$
    copies per circuit execution up to $\log\log$ factors.
    
    Lastly, consider estimating $\tr[G(\rho/\pi) \rho]$ using the Hadamard test and allocating $\varepsilon/(2\alpha_F)$ to the statistical error.
    Note here that $\tr[\beta_F \rho] = \beta_F$, thus we do not need to estimate this term.
    The Hadamard test requires $\mathcal{O}(\alpha_F^2\varepsilon^{-2}\log(\vartheta^{-1}))$ shots.
    Therefore, the total number of copies is
    $\mathcal{O}\!\left(
    D^2 \alpha_F^2\varepsilon^{-2}(\log(D/\varepsilon))^3
    \log(\vartheta^{-1})
    \right)$.
\end{proof}

\section{Results overview}
\label{sec:concept}

In this section, we present our property estimation algorithm in a broad view. 
Particularly, we show the cost for the case when we apply the algorithm to a fundamental nonlinear property of quantum states, von Neumann entropy.
We also present the application to R{\'e}nyi entropy in Appendix~\ref{appendix:renyi}.

\subsection{$(\eta, P)$-spectral cutoff}
First, we formally describe the problem and the main idea.
Suppose that the target function is $f(x)$. 
For the property estimation, the following functions are typical: $\log x$, $x^{c}$ with negative or non-integer $c$, the sign function $\mathrm{sign}(x)$, and the step function $\Theta(x)$. 
These functions are singular or non-analytic at $x=0$.
To approximate such a function, we use an implementable polynomial function $G(x)$ and its modification $F(x) = \alpha_F G(x/\pi) +\beta_F$ satisfying 
\begin{align}
    \left|  G(x/\pi) - \frac{f(x) - \beta_F}{\alpha_F}  \right| &\leq
    \varepsilon_{\mathrm{poly}}, \quad \forall x\in[\delta, 1] \notag,
    \\
    |G(x)| &\le 1, \quad \forall x\in[-1, 1],\label{eq:def-g}
\end{align}
where $\alpha_F > 0$ is a normalization factor ensuring $|G(x)| \le 1$,
$\varepsilon_{\mathrm{poly}} > 0$ is a polynomial approximation error, and $\delta>0$ is an eigenvalue threshold.
Then, $F(x)$ satisfies
\begin{align}
    \left|  F(x) - f(x) \right| &\leq
    \alpha_F\varepsilon_{\mathrm{poly}}, \quad \forall x\in[\delta, 1] \notag,
    \\
    |F(x)| &\le \alpha_F + |\beta_F|, \quad \forall x\in[-1, 1].\label{eq:def-f}
\end{align}
For the typical functions discussed above, the degree of the polynomial $G$ must scale as
$\Theta(\delta^{-1}\log(\varepsilon_{\mathrm{poly}}^{-1}))$ in order to satisfy Eq.~\eqref{eq:def-g}~\cite{gilyen2019-qsvt}.
Hence, the complexity of QSVT to implement such $G$ is dominated by $\delta^{-1}$.
For the target $\tr[f(\rho)\rho]$, decomposing the spectrum into the main approximation region $[\delta,1]$ and the tail region $[-1,\delta]$ gives the following bound on the estimation bias of QSVT:
\begin{eqnarray}
& & \hspace*{-1em}
   \Big| \tr[F(\rho) \rho ] - \tr [f(\rho)\rho ] \Big|
\nonumber \\ & & \hspace*{0.8em}
    = \Big|  \tr [(F - f) \rho \Pi_{\rho > \delta}] 
     + \tr[( F - f) \rho \Pi_{\rho \leq \delta}]  \Big|
\nonumber \\ & & \hspace*{0.8em}
    \le \alpha_F \varepsilon_{\mathrm{poly}} 
       + \Big|\tr[F(\rho) \rho \Pi_{\rho \leq \delta}] \Big|  
       + \Big|\tr[f(\rho) \rho \Pi_{\rho \leq \delta}] \Big|
\nonumber \\ & & \hspace*{0.8em}
\le \alpha_F
    \varepsilon_{\mathrm{poly}} + (\alpha_F+ |\beta_F|) C(\delta) 
       + \Big|\tr[f(\rho) \rho \Pi_{\rho \leq \delta}] \Big|
\nonumber
    \\ & & \hspace*{0.8em}
      := \alpha_F\varepsilon_{\mathrm{poly}} + \varepsilon_{\mathrm{tail}}(\delta).
\label{eq:error_decomposition}
\end{eqnarray} 

Here, $\Pi_{\rho \leq \delta}$ and $\Pi_{\rho > \delta}$ are the projectors onto the subspace spanned by eigenvectors of $\rho$ with eigenvalues $\lambda \leq \delta$ and $\lambda > \delta$, respectively.
The term 
\begin{equation}
    C(x) = \tr[\rho \Pi_{\rho \leq x}]
\end{equation}
represents the cumulative distribution function (CDF) of the eigenvalues. Note that $C(x)$ is an unknown function.
In Eq.~\eqref{eq:error_decomposition}, the total error is bounded by the three components: the polynomial approximation error in the main approximation region, the contribution from the tail region, and the missing true value of the target function in the tail.
We define $\varepsilon_{\mathrm{tail}}(\delta)$ as the tail-induced approximation error from the second and third terms.

In order to lower the complexity for approximating $f$, we seek to find $\delta$ as large as possible such that $\varepsilon_{\mathrm{tail}}(\delta) \le \eta$ for an error tolerance $\eta>0$.
If we choose an appropriate function $P_{\mathrm{tail}}$ satisfying
\begin{align}
    P_{\mathrm{tail}}(\delta) \ge \varepsilon_{\mathrm{tail}}(\delta)\label{eq:bound_P}
\end{align}
for each target $f$, we can control its tail approximation error $\varepsilon_{\mathrm{tail}}(\delta)$ through a cutoff of $P_{\mathrm{tail}}(\delta)$.
For this purpose, we define a \textit{boundary} $(\eta, P)$-spectral cutoff $\lambda_\eta$.
Given a target value $\eta$ and a bounding function $P(\delta)$, 
$\lambda_\eta$ is formally defined as 
\begin{equation}
\label{eq:def-cutoff}
    \lambda_\eta = \sup \{\delta \in(0, 1]: P(\delta) \le \eta \}.
\end{equation}
This definition is analogous to the percentile, i.e.,
$(\eta, P)$-spectral cutoff coincides with the $\eta$-percentile of the density matrix in the case $P(\delta) = C(\delta)$. 
However, for a discontinuous function $P$, this boundary cutoff may not satisfy $P(\lambda_\eta)\le \eta$.
For instance, when $P(x)=0$ for $x\in (0,1/2)$ and $P(x)=1$ otherwise, we have $\lambda_{0}=1/2$ but it gives $P(\lambda_0)=1>0$.
Instead, we introduce a \textit{certified} $(\eta, P)$-spectral cutoff $\hat{\lambda}_{\eta} >0$ satisfying 
\begin{equation}
\label{eq:def-certified-cutoff}
    P(\hat{\lambda}_\eta) \le \eta,
\end{equation}
and we aim to find $\hat{\lambda}_{\eta}$ as close to $\lambda_{\eta}$ as possible.
Using a certified $( \eta, P_{\mathrm{tail}})$-spectral cutoff $\hat{\lambda}_\eta$, the tail-induced error can be bounded as $\varepsilon_{\mathrm{tail}}(\hat{\lambda}_{\eta})\le P_{\mathrm{tail}}(\hat{\lambda}_\eta) \le \eta$. 
We remark that a maximizer in Eq.~\eqref{eq:def-cutoff} need not exist.

As evident from this formulation, the $(\eta, P_{\mathrm{tail}})$-spectral cutoff is an essential parameter that captures the approximation error. 
In contrast, the minimum non-zero eigenvalue $\lambda_{\min}$ and the rank $r$ used in previous works provide merely the worst-case upper bounds of the polynomial degree.
For instance, in the simple case $P_{\mathrm{tail}} = C$, one can trivially ensure $C(\delta) \le \eta$ by setting $\delta < \lambda_{\min}$ (yielding $C(\delta)=0$) or $\delta \le \eta/r$.
Thus, these inputs also provide the worst-case lower bound of $\lambda_\eta$.
However, these bounds are often loose and, as a result, the required polynomial degree in QSVT can be overestimated; see Fig.~\ref{fig:idea}. 

By construction, both boundary and certified $(\eta, P_{\mathrm{tail}})$-spectral cutoffs $\lambda_{\eta}$ and $\hat{\lambda}_{\eta}$ depend on the task, the state, and the target accuracy.
To make full use of its theoretical advantages, a $(\eta, P_{\mathrm{tail}})$-spectral cutoff should be carefully determined. 
Once $\eta$ and $P_{\mathrm{tail}}$ are fixed, the remaining task is to identify an appropriate cutoff $\hat{\lambda}_\eta$ satisfying Eq.~\eqref{eq:def-certified-cutoff}, depending on the spectrum of $\rho$. 
This state dependence is the main advantage of our method. However, because $\rho$ is unknown, identification of $\hat{\lambda}_\eta$ is challenging in practice, even in the simplest case $P_{\mathrm{tail}} = C$. 

\subsection{Two-stage framework}\label{sec:two-stage}

Our property estimation algorithm is composed of two stages. 
The first stage deals with the difficulty in identifying the state-dependent certified cutoff $\hat{\lambda}_\eta$ via an iterative method. 
In the second stage, we estimate the property using a low-degree polynomial designed from the identified cutoff $\hat{\lambda}_\eta$ in the first stage. 

The overall pipeline is illustrated in Fig.~\ref{fig:pipeline}. 
The primary inputs are multiple copies of $\rho$, the accuracy $\varepsilon$ for the property estimation, and the failure probability $\vartheta$.
Given these inputs, the algorithm proceeds as follows.
\begin{description}
    \item[Preparation] Before the execution on a quantum device, we design the algorithm parameters: the target value $\eta$ and the bounding function $P$,
    so that the property is estimated within the additive error $\varepsilon$.
    
    \item[Stage~1 (Identification)] Given $\eta, P, \vartheta$ and $\rho$, a quantum computer executes Algorithm~\ref{alg-search}.
    The algorithm identifies a certified $(\eta, P)$-spectral cutoff $\hat{\lambda}_\eta$ with a success probability at least $1- \vartheta/2$. 
    This stage consumes $\tilde{\mathcal{O}}(\hat{\lambda}^{-2}_\eta\eta^{-2}\log(\vartheta^{-1}))$
    copies of $\rho$ as stated in Theorem~\ref{thm:search} in Section~\ref{sec:search} in the case $P=C$.
    A lower bound of $\hat{\lambda}_\eta$ is also guaranteed in Theorem~\ref{thm:search}.
    
    \item[Stage~2 (Estimation)] Using the identified cutoff $\hat{\lambda}_\eta$, the polynomials $F$ and $G$ are classically designed with degree given by $D = \tilde{\mathcal{O}}(\hat{\lambda}_\eta^{-1})$.
    Given $F$, $\varepsilon$, $\vartheta$, and $\rho$, a quantum computer executes a QSVT-based algorithm, which estimates $\tr[F(\rho)\rho]$ with a success probability at least $1- \vartheta/2$ as shown in Lemma~\ref{lemma:virtual-qsvt}. 
    This stage gives an estimate of the target property $\mathrm{tr}[f(\rho)\rho]$, by consuming $\tilde{\mathcal{O}}(\hat{\lambda}_\eta^{-2}\alpha_F^2\varepsilon^{-2}\log(\vartheta^{-1}))$ copies of $\rho$.
\end{description}
\begin{figure}[htbp]
    \centering    
    \includegraphics[width=1\linewidth]{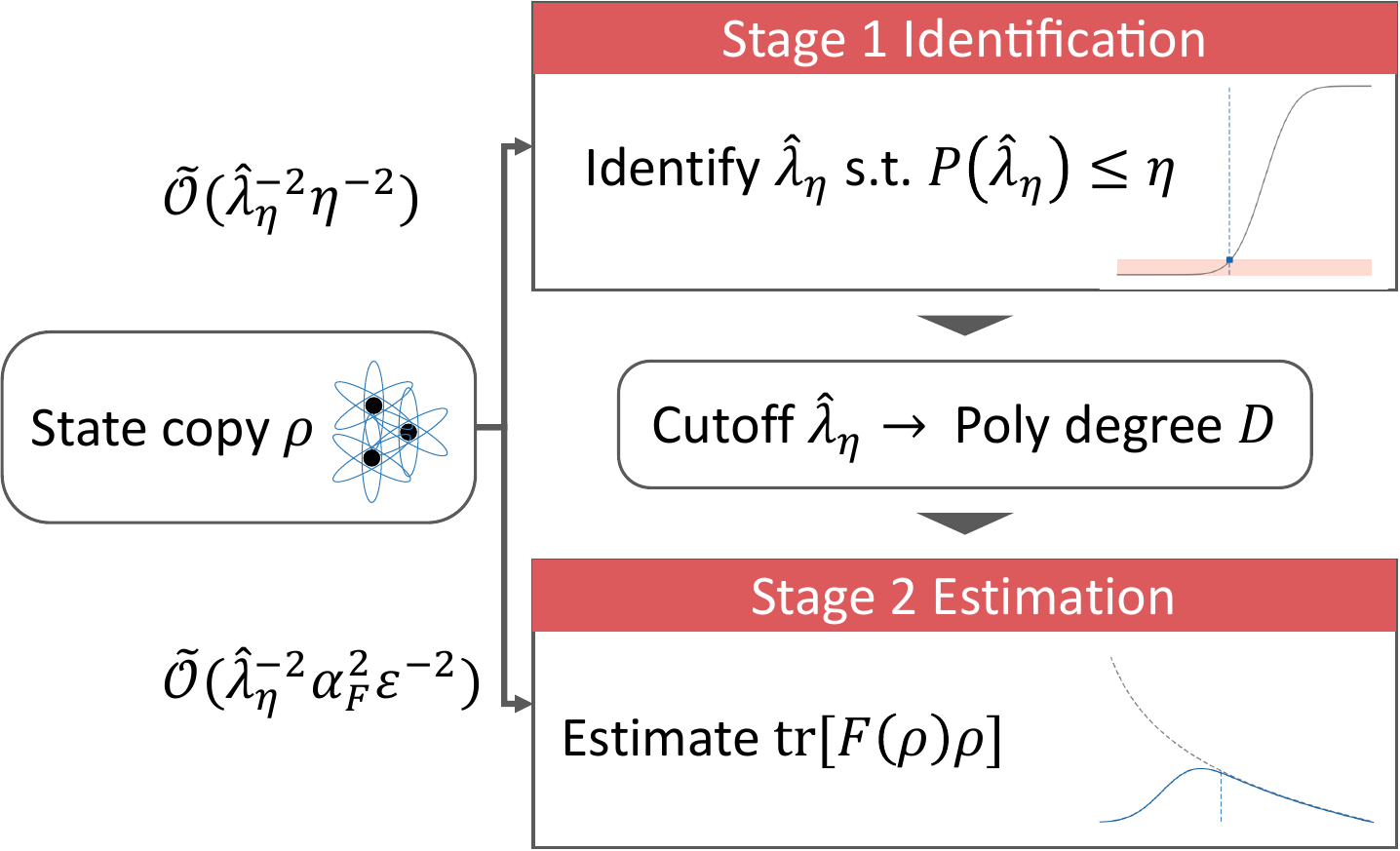}
    \caption{The overall pipeline of the two-stage algorithm.}
    \label{fig:pipeline}
\end{figure}

We apply this framework for estimating von Neumann entropy and evaluate the total copy requirements, summarized in Table~\ref{tab:total-complexity}; the proof will be given in Lemma~\ref{lemma:vn}.
In this case, as will be shown later, we can choose $P=C$. 
Notably, the copy requirement in Stage~1 is comparable to or smaller than that in Stage~2.
This implies that Stage~1 does not change the leading-order of the total copy complexity.
Furthermore, while the complexity of our algorithm is determined by $\hat{\lambda}_\eta$,
we also derive a worst-case complexity based on the rank $r$ of $\rho$ or $\lambda_{\min}$, for the fair comparison to previous works. 
More precisely, from Theorem~\ref{thm:search} ($\hat{\lambda}_\eta$ is larger than $\Omega(\eta/r)$ or $\Omega(\lambda_{\min})$ without knowing $r$ or $\lambda_{\min}$) and $\eta\sim\varepsilon$,
our worst-case copy requirements are
$\tilde{\mathcal{O}}(r^2 \varepsilon^{-4})$ and $\tilde{\mathcal{O}}(\lambda_{\min}^{-2} \varepsilon^{-2})$, which match the best-known scaling in the sample-access model~\cite{Kean2025-list, wada2025state}.
As illustrated in Section~\ref{sec:case-study-demonstration}, since the identified cutoff $\hat{\lambda}_\eta$ can be substantially larger than the worst-case bound, our algorithm can practically improve the total copy requirement for von Neumann entropy estimation.

\begin{table}[htbp]
    \centering
    \renewcommand{\arraystretch}{1.5}
    \begin{tabular}{ccc|c}\hline \hline
        Identification ~~ & Estimation ~~ & Total ~~ & ~~ Total (worst case)\\
        \hline
        $\hat{\lambda}_{\eta}^{-2} \varepsilon^{-2}$ & $\hat{\lambda}_{\eta}^{-2} \varepsilon^{-2}$ &
        $\hat{\lambda}_{\eta}^{-2} \varepsilon^{-2}$ &
        $r^2 \varepsilon^{-4}, ~~~ \lambda_{\min}^{-2} \varepsilon^{-2}$\\
        \hline \hline
    \end{tabular}
    \renewcommand{\arraystretch}{1}
    \caption{
    The leading order of necessary copies of $\rho$ for estimating von Neumann entropy within an additive error $\varepsilon$ with probability at least $2/3$; the copies are consumed in the identification stage (Stage~1) and the estimation stage (Stage~2). 
    The worst-case copy requirements described in terms of the rank $r$ or the minimum non-zero eigenvalue $\lambda_{\min}$ match the best-known scaling~\cite{Kean2025-list, wada2025state}. 
    }
    \label{tab:total-complexity}
\end{table}

\section{Identification algorithm for the spectral cutoff}
\label{sec:search} 

In this section, we present an iterative algorithm for identifying $\hat{\lambda}_\eta$ 
(Algorithm~\ref{alg-search}) in Stage~1. 
The algorithm presented here employs the core case of 
\begin{equation}
   P(\delta) = C(\delta), 
\end{equation}
which is applicable to estimating von Neumann entropy.
In Appendix~\ref{appendix:renyi}, we provide the identification algorithm for another choice of $P$ applicable to R{\'e}nyi entropy. 

The key idea is to construct a degree-$d$ approximate CDF $\tilde{C}_d(x)$ and find the largest $x$ satisfying $\tilde{C}_d(x)\leq\eta$ via a classical search, because the exact CDF $C(x)$ is unknown. 
To construct $\tilde{C}_d(x)$, we leverage the efficient Fourier-based approach developed in Refs.~\cite{lin2022heisenberg, wan2022randomized}.
These previous works aimed to estimate the minimum non-zero eigenvalue of a given Hamiltonian, under the assumption that the prepared state has a large overlap with the true ground state.
We leverage the CDF construction and modify it for our search purpose.

\subsection{Approximate CDF construction}

Here, we review the method for constructing the degree-$d$ approximate CDF $\tilde{C}_d(x)$~\cite{lin2022heisenberg, wan2022randomized}.
Let $\rho$ be a quantum state and
$H$ be a Hamiltonian.
For validating our search algorithm, we assume $0 \le H \le I$, while the original CDF construction works if $\|H\| \le \pi /3$. 
Later, we will apply the result to our primary interesting case $H=\rho$.
Let $p$ denote the spectral distribution of $H$ induced by the input state $\rho$; 
that is, $p(x)=\sum_k p_k \delta(x-\lambda_k)$ with $\lambda_k$ the eigenvalue of $H$ and $p_k$ the probability to get $\lambda_k$ when measuring $H$.
Following Refs.~\cite{lin2022heisenberg, wan2022randomized}, the exact CDF is written as
the convolution of the Heaviside step function $\Theta(x)$ and the spectral distribution:
\begin{equation}
    C(x)=(\Theta*p)(x)=\tr[\rho\Pi_{H\le x}],
\end{equation}
where $\Pi_{H \le x}$ is the projector onto the eigenspaces of $H$ with all eigenvalues $\lambda_k$ satisfying $\lambda_k \leq x$. 
By approximating the Heaviside step function with a truncated Fourier series as 
\begin{equation}
    \Theta(x) \approx \tilde{\Theta}(x) = \sum_{|j| \le d} F_{j,d} e^{i x j}, 
\end{equation}
we obtain the degree-$d$ approximate CDF $\tilde{C}_d(x)$:
\begin{equation}
    \tilde{C}_d(x) = (\tilde{\Theta}*p)(x)
    = \sum_{|j| \le d} F_{j,d} e^{i x j} \tr[\rho e^{-i H j}],
\label{eq:def-tilde-Cd}
\end{equation}
where $d$ is a positive odd integer and $F_{j,d}$ are the Fourier coefficients explicitly defined in Appendix~\ref{appendix:fourier-heaviside}. 
Note that 
\begin{equation}
    \mu_j = \tr[\rho e^{-i H j}]
\end{equation}
can be estimated via the Hadamard test.
The approximate CDF $\tilde{C}_d(x)$ satisfies the following error bound:
\begin{equation}\label{eq:tildeC-bound}
    C(x - \delta) -  \zeta \leq  \tilde{C}_d(x) \leq C(x + \delta) + \zeta,
\end{equation}
where $\delta \in (0, \pi/2)$ denotes a resolution of the approximate CDF, and $\zeta > 0$ is an approximation error of $\Theta$, i.e., $|\Theta(x)-\tilde{\Theta}(x)|\leq \zeta$ in $x\in [-1, -\delta] \cup [\delta, 1]$.
The degree $d$ depends on these parameters as $d = g(\delta, \zeta) = \mathcal{O}\!\left(\delta^{-1} \log(\zeta^{-1})\right)$, where $g(\delta, \zeta)$ is formally defined in Appendix~\ref{appendix:fourier-heaviside}.

To estimate $\tilde{C}_d(x)$, the previous works \cite{lin2022heisenberg, wan2022randomized} take an importance sampling approach.
To implement the importance sampling, we rewrite $\tilde{C}_d(x)$ as
\begin{align}
\label{eq:approx-cdf-with-mu}
    \tilde{C}_d(x)
    = \mathcal{F}_d \sum_{|j|\le d} e^{i \phi_{j,d}(x)} \frac{|F_{j,d}|}{\mathcal{F}_d} \mu_j,
\end{align}
where $\phi_{j,d}(x)$ is defined such that $F_{j,d}e^{ixj} = |F_{j,d}| e^{i \phi_{j,d}(x)}$ and $\mathcal{F}_d = \sum_{|j|\le d} |F_{j,d}|$.
Let $J$ be a random variable such that $\mathrm{Pr}[J=j] = |F_{j,d}|/\mathcal{F}_d$.
Using $N_{d}$ samples of $J$, we obtain $\hat{C}_d(x)$ as an estimate of $\tilde{C}_d(x)$:
\begin{align}
    \hat{C}_d(x)
    &= \frac{\mathcal{F}_d}{N_{d}}\sum_{k=1}^{N_{d}} \mathrm{Re}[e^{i\phi_{J_k,d}(x)}(X_{J_k} + i Y_{J_k})], \nonumber\\
    &= \frac{\mathcal{F}_d}{N_{d}}\sum_{k=1}^{N_{d}} \left(X_{J_k}\cos \phi_{J_k,d}(x) -  Y_{J_k}\sin \phi_{J_k,d}(x) \right), 
    \label{eq:def-hat-c}
\end{align}
where $J_k$ is the $k$th sample of $J$, and $X_{J_k}, Y_{J_k} \in \{\pm 1\}$
are outcomes obtained from Hadamard tests for the real and imaginary parts of $\mu_{J_k}=\tr[\rho e^{- iH J_k}]$, respectively.
This definition ensures that the estimator is real-valued for every finite sample size, and $\mathbb{E}[\hat{C}_d(x)] = \tilde{C}_d(x)$.
Alternatively, one may use symmetrized sampling for Fourier modes $j$ and $-j$ to construct the real-valued estimator.

\subsection{Iterative search algorithm}

Now let us consider the problem of identifying a certified $(\eta, C)$-spectral cutoff $\hat{\lambda}_{\eta}$.
There are two main challenges in our setting.
First, as in the previous CDF-based search algorithms, we cannot use the exact unknown CDF $C(x)$ but only its estimate $\hat{C}(x)$ to identify $\hat{\lambda}_{\eta}$. 
We need to guarantee $C(\hat{\lambda}_\eta)\leq \eta$ even under Fourier-based approximation and statistical errors contained in $\hat{\lambda}_\eta$. 
Second, an appropriate choice of the resolution width $\delta$, which largely affects both the estimation cost and the quality of the cutoff $\hat{\lambda}_\eta$, is unknown a priori.

To address the first challenge, we provide a certification rule with high probability.
For given resolution width $\delta$, Fourier approximation error $\zeta$, 
statistical error $\xi$,
algorithm error arising from Hamiltonian evolution $\kappa$,
and an evaluation point $x$, we choose $d = g(\delta, \zeta)$ and $N_d$ so that the following relation holds with a success probability at least $1-\vartheta$:
\begin{gather}
        C(x - \delta) - \zeta - \xi -\kappa \leq  \hat{C}_d(x) \leq C(x + \delta) + \zeta + \xi + \kappa, \notag\\
        C(x - \delta) - E \leq  \hat{C}_d(x) \leq C(x + \delta) +E,
    \label{eq:cdf-approximation-bound}
\end{gather}
where $E = \zeta + \xi + \kappa$ is an error margin.
To satisfy this condition,
\begin{equation}
    N_d = \tilde{\mathcal{O}}(\mathcal{F}_d^2 \xi^{-2} \log(\vartheta^{-1}))
\end{equation}
is required.
Then, from $\hat{C}_d(x)$, we obtain an eigenvalue candidate
\begin{equation}
\label{eq:find-lambda-star}
     \lambda_{*} = \underset{x\in \chi_{\delta}:\hat{C}_d(x + \delta) \leq
     \eta-E}{\max} x,
\end{equation}
where $\chi_{\delta}$ is the following set of grid points:
\begin{equation}
    \chi_{\delta} =\{0, \delta, 2\delta, \cdots, \lfloor
    \overline{\| H \|}/\delta\rfloor \delta \},
     \label{eq:chi}
\end{equation}
with the upper bound $\overline{\| H \|}$ of $\| H \|$.
If the condition $\hat{C}_d(x + \delta) \leq  \eta - E$ is not satisfied, we set $\lambda_*=0$. 
Then we can certify that $C(x)\le \eta$ because 
\begin{equation}\label{eq:certification-cond}
    C(x)  \le \hat{C}_d(x +  \delta) + E \leq \eta.
\end{equation}
Thus, the candidate $\lambda_*$ is a certified cutoff.

To address the second challenge, we refine the resolution geometrically as
$\delta_l = 2^{-l}$ with a positive integer $l$, or equivalently increase the approximation degree $d_l$.
This refinement schedule does not introduce an additional asymptotic overhead.
To ensure that the identified cutoff candidate $\lambda_*$ is well-separated from zero (that is, a nontrivial $\lambda_*$ is obtained), we terminate the algorithm when $\lambda_* \ge 2\delta_l$ holds, and then return $\hat{\lambda}_\eta =2\delta_l$.
In practice, $\hat{\lambda}_\eta$ may be chosen by $\hat{\lambda}_\eta = \lambda_* \ge 2\delta_l$ to reduce the polynomial degree, but here we take $\hat{\lambda}_\eta = 2\delta_l$ to simplify the following complexity analysis.

    We remark on the existence of a certified cutoff.
    Since we aim to find an eigenvalue threshold, we are interested in a positive certified cutoff $\hat{\lambda}_\eta > 0$.
    Let $C(0) = \tr[\rho \Pi_{H=0}]$ be a population of $\rho$ in the kernel of $H$, where $\Pi_{H=0}$ is a projector onto the kernel of $H$.
    For a positive semidefinite Hamiltonian $H$, a positive certified cutoff $\hat{\lambda}_\eta$ exists if  $C(0) \le \eta$.
    In our algorithm, we assume $C(0) \le \eta/2$ with an error margin for the search.
    In the primary case $H=\rho$, the assumption is directly satisfied since $C(0)=0$.

Algorithm~\ref{alg-search} is our identification algorithm for $\hat{\lambda}_\eta$ that incorporates the above ideas.
Lemma~\ref{lemma:search} below shows the complexity of Algorithm~\ref{alg-search}
for Hamiltonian $H$. 
Theorem~\ref{thm:search} shows the complexity in the primary case $H=\rho$, by replacing $\hat{C}_d$ by 
$\hat{C}^{(v)}_d$ in Algorithm~\ref{alg-search}, and setting the error allocation $\zeta=\xi=\kappa = \eta/12$.
The total copy requirement of the algorithm $\tilde{\mathcal{O}}(\hat{\lambda}_\eta^{-2}\eta^{-2})$ is comparable to that of Stage~2. 
Also, Lemma~\ref{lemma:search} and Theorem~\ref{thm:search}  guarantee the termination of the algorithm and provide a lower bound of the algorithm output.
\noindent
\begin{minipage}{0.95\columnwidth}
    \begin{algorithm}[H]
    \caption{Iterative search of spectral cutoff}
    \label{alg-search}
    \noindent
        \begin{algorithmic}[1]
        \Require {Access to $\rho$ and controlled $e^{-iH}$, target value $\eta \in (0, 1)$ satisfying $C(0)\le \eta /2$ , success probability $1-\vartheta \in ( 0,1)$.  }
        \Ensure A certified cutoff $\hat{\lambda}_{\eta}$ such that $C(\hat{\lambda}_{\eta}) \le \eta$ \\ with success probability at least $1-\vartheta$.
        \State Initialize fixed parameters: $\xi \leftarrow \eta / 8$, $\zeta \leftarrow \eta / 8$, $\kappa \leftarrow 0$, and $E \leftarrow \zeta + \xi + \kappa$.
        \For{$l=1$ to $\infty$}
            \State{  Set variables: $\delta_{l} \leftarrow 2^{-l}$, $d_l \leftarrow         g(\delta_l, \zeta)$,
                $\chi_{\delta_l}\leftarrow$Eq.~\eqref{eq:chi}, 
                $\vartheta_{l} \leftarrow 6 \vartheta / (\pi^2 l^2)$, and 
                    $N_{d_l}\leftarrow$Eq.~\eqref{eq:ndl}.
                
                }
            \State {Construct $\hat{C}_{d_l}(x) \leftarrow$ Eq.~\eqref{eq:def-hat-c} by $2N_{d_l}$ measurement shots.}
            \State {Find $\lambda_{*} \leftarrow$ Eq.~\eqref{eq:find-lambda-star} via classical search on $\chi_{\delta_l}$ \\
            with $\delta\leftarrow \delta_l$ and $d\leftarrow d_l$. }
            \If{$\lambda_{*} \ge  2 \delta_{l}$}
                \State \Return $\hat{\lambda}_{\eta } =  2\delta_l$
            \EndIf
        \EndFor
        \end{algorithmic}
    \end{algorithm}
\end{minipage}

\begin{lemma}[Spectral cutoff search for Hamiltonian evolution access]\label{lemma:search}
    Let $\rho$ be an unknown but accessible quantum state.
    Let $H$ be a positive semidefinite Hamiltonian with $\|H\| \leq 1$, and suppose we have access to controlled Hamiltonian evolution $e^{-i H}$.
    Let $\eta \in (0,1)$ be a target value, and $\vartheta \in (0,1)$ be a failure probability.
    Assume that $C(0) = \tr[\rho \Pi_{H=0}]$ satisfies $C(0) \le \eta/2$.
    For $\rho$, $\eta$, and $\vartheta$, Algorithm~\ref{alg-search} identifies a certified $(\eta, C)$-spectral cutoff $\hat{\lambda}_\eta$ with a success probability at least $1-\vartheta$ such that $C(\hat{\lambda}_\eta) \le \eta$ with
    $\tilde{\mathcal{O}} (\hat{\lambda}_{\eta}^{-1})$ simulation time in a single circuit, and $\tilde{\mathcal{O}}\!\left(\hat{\lambda}_{\eta}^{-1} \eta^{-2}\log(\vartheta^{-1})\right)$ total simulation time.
    
    Moreover, Algorithm~\ref{alg-search} terminates before the resolution $\delta$ becomes below
    $\lambda_{\eta/2}/8$, i.e., the returned cutoff $\hat{\lambda}_\eta$ satisfies
    \begin{equation}
        \hat{\lambda}_{\eta}  \ge \frac{\lambda_{\eta/2}}{4},
    \end{equation}
    where $\lambda_{\eta/2}$ denotes 
    the boundary $(\eta/2, C)$-spectral cutoff defined in Eq.~\eqref{eq:def-cutoff}.
    
\end{lemma}
\begin{proof}
    Algorithm~\ref{alg-search} iteratively constructs $\hat{C}_{d_l}(x)$ with decreasing $\delta_{l}=2^{-l}$, or equivalently increasing $d_l = g(\delta_l, \zeta)$.
    We define $E=\zeta+\xi+\kappa$.
    As we confirm below, the algorithm terminates if $E \le \eta/4$.
    In the present lemma, Hamiltonian evolution is exact,
    so we set $\kappa=0$ and choose $\xi=\zeta=\eta/8$.
    
    For each $l$, given $\delta_l$, $\vartheta_{l}$ and $d_l$,
    we first confirm Eq.~\eqref{eq:cdf-approximation-bound} holds.
    We start from Eq.~\eqref{eq:tildeC-bound}.
    To ensure $|\tilde{C}_d(x) - \hat{C}_d(x)| \le \xi$ with a high probability,
    we observe that
    a sufficient number of measurement shots is 
    $\left\lceil\frac{ 4\mathcal{F}_{d_l}^2}{\xi^2}  \log(2\vartheta_{l}^{-1}) \right\rceil$
    for estimating the real or imaginary part
    by using $|X_{J_k}\cos \phi_{J_k,d}(x) -  Y_{J_k}\sin \phi_{J_k,d}(x)| \le
    |\cos \phi_{J_k,d}(x)| +  |\sin \phi_{J_k,d}(x)| \le \sqrt{2}$ and 
    applying Hoeffding's inequality to Eq.~\eqref{eq:def-hat-c}.
    For the evaluation on all grid points in $\chi_{\delta_l}$
    with a success probability at least $1-\vartheta_{l}$,
    \begin{equation}\label{eq:ndl}
        N_{d_l} = \left\lceil\frac{ 4\mathcal{F}_{d_l}^2}{\xi^2}  \log(2|\chi_{\delta_l}|\vartheta^{-1}_l) \right\rceil
    \end{equation}
    is sufficient using the union bound over all grid points. 
    To estimate both the real and imaginary parts, we need $2N_d$ measurement shots.
    Thus, by evaluating the certification condition~\eqref{eq:certification-cond} on all grid points,  we can find a candidate $\lambda_*$ satisfying $C(\lambda_*) \le \eta$.

    The algorithm terminates when the condition $\lambda_* \ge 2\delta_l$ is met, meaning that, for an unknown $L$ at the termination,
    $\hat{\lambda}_\eta = 2\delta_L$ holds.
    Although the terminal iteration $L$ is not known in advance, we can allocate the failure probability over all possible iterations by setting
    \begin{equation}
        \vartheta_l = \frac{6\vartheta}{\pi^2 l^2}.
    \end{equation}
    This ensures that the total failure probability is at most $\vartheta$ since $\sum_{l=1}^\infty \vartheta_l\le \vartheta$.
    Therefore, we obtain $\hat{\lambda}_\eta$ with a success probability at least $1-\vartheta$.
    Note that the algorithm must terminate because the existence of such a finite $L$ is established below.
    
    Next, we summarize the complexities.
    The maximum simulation time in a single circuit is given by $d_L = \tilde{\mathcal{O}}(\delta_L^{-1})$.
    For a fixed $l$, the total simulation time is given by
    $\sum_{k=1}^{N_{d_l}} |J_k|$.
    The total simulation time is stochastic, but we can provide a deterministic upper bound as follows:
    \begin{align}
        2\sum_{k=1}^{N_{d_l}} |J_k|
        & \le 2N_{d_l} d_l \notag \\
        &= \tilde{\mathcal{O}}(\mathcal{F}_{d_l}^2 \xi^{-2} \log(|\chi_{\delta_l}|\vartheta_l^{-1})d_l)\notag \\
        &= \tilde{\mathcal{O}}(d_l\xi^{-2} (\log(d_l))^2\log(l\delta_l^{-1}\vartheta^{-1})) \notag \\ &=\tilde{\mathcal{O}}(\delta_l^{-1} \eta^{-2}\log(\vartheta^{-1})),
    \end{align}
    where we use $d_l = \tilde
    {\mathcal{O}}(\delta_l^{-1})$
    and $\mathcal{F}_{d_l} = \mathcal{O}(\log(d_l))$~\cite{lin2022heisenberg, wan2022randomized}.
    Hence, the total simulation time is 
    \begin{align}
        2\sum_{l=1}^{L}\sum_{k=1}^{N_{d_l}} |J_k|
        &=\tilde{\mathcal{O}}(\delta_L^{-1} \eta^{-2}\log(\vartheta^{-1}))\notag\\
        &=\tilde{\mathcal{O}}(\hat{\lambda}_\eta^{-1} \eta^{-2}\log(\vartheta^{-1})),
    \end{align}
    where we take $\delta_l$ as $\delta_l = 2^{-l}$ and we have $\delta_L^{-1}  = \Theta(\hat{\lambda}_\eta^{-1})$ since $2\delta_L  = \hat{\lambda}_\eta$.

    In the latter part, we prove the lower bound of the cutoff.
        The assumption $C(0) \le  \eta/2$ guarantees $\lambda_{\eta/2}>0$. 
        Thus, there exists a dyadic resolution $\delta_{\text{ test}}= 2^{-l} $ for a positive integer $l$ with
        $\lambda_{\eta/2}/8\le \delta_{\rm test}<\lambda_{\eta/2}/4$.
        We show that $\delta_{\rm test}$ satisfies the stopping condition.
        The algorithm may stop earlier, but it cannot continue beyond this resolution $\delta_{\text{ test}}$.
    
        Consider the test point $x_{\text{test}} = 2\delta_{\text{test}}$.
        Since $x_{\rm test}+2\delta_{\rm test}<\lambda_{\eta/2}$, the monotonicity of $C$ and the definition of $\lambda_{\eta/2}$ as a supremum imply $C(x_{\rm test}+2\delta_{\rm test})\le \eta/2$.
        Therefore, the approximate CDF is bounded as
        \begin{align}
            \hat{C}_d(x_{\text{test}} + \delta_{\text{test}})
        &\le C(x_{\text{test}} + 2\delta_{\text{test}})  +E \notag \\
        &\le \frac{\eta}{2}  +E \le \eta  -E,
        \end{align}
        where the last inequality follows from the condition $E \le \eta/4$.
        Thus, $x_{\text{test}}$ is a valid candidate of the search.
        Since $\lambda_*$ is chosen as the largest valid candidate, $\lambda_*$ satisfies the termination condition $\lambda_* \ge x_{\text{test}} =  2\delta_{\text{test}}$.
        Hence, the algorithm must terminate no later than the iteration with resolution $\delta_{\text{test}}$.
        If the algorithm terminates before reaching $\delta_{\text{test}}$, the stopping condition implies that the returned cutoff is at least $2\delta_{\text{test}}$.
        Otherwise, at the iteration with resolution $\delta_{\text{test}}$, we have $\hat{\lambda}_{\eta} = 2\delta_{\text{test}} \ge \lambda_{\eta/2} /4$.
    
\end{proof}

\begin{thm}[Spectral cutoff search for sample access]\label{thm:search}

    Let $\rho$ be an unknown but accessible quantum state
    and consider the spectral cutoff search with $H=\rho$.
    Let $\eta \in (0,1)$ be a target value, and $\vartheta \in (0,1)$ be a failure probability.
    For $\rho$, $\eta$, and $\vartheta$,
    there exists an algorithm that identifies the certified $(\eta, C)$-spectral cutoff $\hat{\lambda}_\eta$ with a success probability at least $1-\vartheta$ such that $C(\hat{\lambda}_\eta) \le \eta$.
    This algorithm uses
    $\tilde{\mathcal{O}}(\hat{\lambda}_{\eta}^{-2})$ copies of $\rho$ in a single circuit and $\tilde{\mathcal{O}}\!\left(\hat{\lambda}_{\eta}^{-2} \eta^{-2}\log(\vartheta^{-1})\right)$ total copies.

    Moreover, the returned cutoff $\hat{\lambda}_\eta$ satisfies
    \begin{equation}
        \hat{\lambda}_{\eta} \ge
        \frac{\lambda_{\eta/2}}{4} \ge
        \max \left\{\frac{1}{4} \lambda_{\min}, \frac{\eta}{8r} \right\},
    \end{equation}
    where $\lambda_{\min} > 0$ is the minimum non-zero eigenvalue and $r$ is the rank of $\rho$.

\end{thm}
\begin{proof}

    We modify Lemma~\ref{lemma:search} for the virtual DME.
    First, we evaluate the error margin.
    Let $\hat{C}^{(v)}_d(x)$ be an implementation of $\hat{C}_d(x)$ using virtual DME as
    \begin{equation}\label{eq:cbyvdme}
        \hat{C}_d^{(v)}(x)
        = \frac{\mathcal{F}_d}{N_{d}}\sum_{k=1}^{N_{d}}  \gamma_{k}\mathrm{Re}[e^{i\phi_{J_k,d}(x)}(X^{(v)}_{J_k} + i Y^{(v)}_{J_k})],
    \end{equation}
    with the measurement outcomes $X^{(v)}_{J_k}$ and $Y^{(v)}_{J_k}$ via the virtual DME.
    $\gamma_{k} \in (1, e^2]$ is a rescaling constant of the virtual DME~\cite{wada2025state}.
    The error is decomposed as
    \begin{multline}
        |\hat{C}_{d_l}^{(v)} (x) -\tilde{C}_{d_l} (x)| \\
        \le |\hat{C}_{d_l}^{(v)} (x) -\mathbb{E} [\hat{C}_{d_l}^{(v)} (x)]|  + |\mathbb{E} [\hat{C}_{d_l}^{(v)} (x)] - \tilde{C}_{d_l} (x)|.
    \end{multline}
    Suppose that, at iteration $l$, the controlled virtual DME implementation of
    $e^{-i\rho J_k}$ has diamond-distance error at most $\varepsilon_{\diamond,l}$.
    Then, the bias of CDF is described as
    \begin{equation}
        |\mathbb{E} [\hat{C}_{d_l}^{(v)} (x)] - \tilde{C}_{d_l} (x)|
        \le 
        \frac{\mathcal{F}_{d_l}}{N_{d_l}}\sum_{k=1}^{N_{d_l}} \sqrt{2} \varepsilon_{\diamond,l}
        \le \sqrt{2} \mathcal{F}_{d_l}  \varepsilon_{\diamond,l},
    \end{equation}
    where we use $|\mathbb{E}[\gamma_{k} X_{J_k}^{(v)}] - \mathrm{Re}[\mu_{J_k}]| \le \varepsilon_{\diamond,l}$
    and $|\mathbb{E}[\gamma_kY^{(v)}_{J_k}] - \mathrm{Im}[\mu_{J_k}]| \le \varepsilon_{\diamond,l}$.
    Thus, we choose $\varepsilon_{\diamond,l} = \kappa / (\sqrt{2} \mathcal{F}_{d_l})$ to suppress the bias below $\kappa$.
    $|\hat{C}_{d_l}^{(v)} (x) -\mathbb{E} [\hat{C}_{d_l}^{(v)} (x)]|$ can be evaluated by Hoeffding's inequality.
    Since $\gamma_{k}$ is upper bounded by $e^2$,
    $N_{d_l} = \left\lceil\frac{4e^4\mathcal{F}_{d_l}^2}{\xi^2}  \log(2|\chi_{\delta_l}|\vartheta^{-1}_l) \right\rceil$ is sufficient to suppress the statistical error within $\xi$.
    Hence, we obtain $C(x - \delta_l) - E \leq  \hat{C}^{(v)}_{d_l}(x) \leq C(x + \delta_l) +E$
    for all $x\in\chi_{\delta_l}$ with probability at least $1 - \vartheta_l$.

    We describe Algorithm~1 with $\hat{C}_{d_l}^{(v)}(x)$ in place of $\hat{C}_{d_l}(x)$.
    Accordingly, we modify the certification condition to
    $\hat{C}^{(v)}_{d_l}(x + \delta_l)  \leq  \eta- E$,
    and set $\xi=\zeta=\kappa=\eta/12$ so that $E \le \eta/ 4$.
    The maximum copy requirement of a single circuit is
    $\mathcal{O}\!\left(d_L^2 \log(d_L\varepsilon_{\diamond, L}^{-1})\right) = \tilde{\mathcal{O}}(\hat{\lambda}_\eta^{-2})$.
    Furthermore, the copy requirement for a fixed iteration $l$ is deterministically bounded by
    \begin{align}
        2\sum_{\substack{k=1 \\ J_k \neq0}}^{N_{d_l}} |J_k|^2 \log(|J_k|\varepsilon_{\diamond, l}^{-1})
        &\le 
        2N_{d_l} |d_l|^2 \log(d_l\varepsilon_{\diamond, l}^{-1}) \notag \\
        &= \tilde{\mathcal{O}}(\delta_l^{-2} \eta^{-2}\log(\vartheta^{-1})).
    \end{align}
    Summing up to $L$, the total copy requirement is 
    \begin{align}
    2\sum_{l=1}^L\sum_{\substack{k=1 \\ J_k \neq0}}^{N_{d_l}} |J_k|^2 \log(|J_k|\varepsilon_{\diamond,l}^{-1})
    & \le 2\sum_{l=1}^LN_{d_l} |d_l|^2 \log(d_l\varepsilon_{\diamond, l}^{-1}) \notag \\
    & =\tilde{\mathcal{O}}(\hat{\lambda}_\eta^{-2} \eta^{-2}\log(\vartheta^{-1})).
    \end{align}

    Since $\lambda_{\eta/2} \ge \lambda_{\min}$ and $\lambda_{\eta/2} \ge \eta/(2r)$,
    the worst-case bounds follow directly.

\end{proof}

Lemma~\ref{lemma:search} clarifies the quality of the certified cutoff returned by Algorithm~\ref{alg-search}.
The algorithm output $\hat{\lambda}_\eta$ is a certified cutoff, and is not necessarily the tightest boundary cutoff; that is, $\hat{\lambda}_\eta = \lambda_{\eta}$ may not hold in general.
This gap comes from the search threshold with a margin $E$ for the algorithmic errors $\zeta$, $\xi$, and $\kappa$. 
Nevertheless,  Lemma~\ref{lemma:search} guarantees the upper and lower bounds:
\begin{equation}
    \frac{\lambda_{\eta/2}}{4} \le \hat{\lambda}_\eta \le  \lambda_{\eta}. 
\end{equation}
Moreover, for the primary case $H=\rho$, the lower bound recovers the conventional worst-case cutoff based on $\lambda_{\min}$ and $\eta/r$ up to constants. 
This means that Algorithm~\ref{alg-search} can find a larger cutoff than the conventional worst-case, even if we do not know $\lambda_{\min}$ or $\eta/r$. 

In addition, the final CDF estimate provides a posteriori certificate on the tightness of the returned cutoff $\hat{\lambda}_{\eta}$.
According to Eq.~\eqref{eq:cdf-approximation-bound}, the true CDF $C(\hat{\lambda}_\eta)$ at the cutoff $\hat{\lambda}_{\eta}$ is bounded by $\eta' \le C(\hat{\lambda}_\eta) \le \eta$, where
$\eta' = \max\{\hat{C}_d(\hat{\lambda}_\eta - \delta) - E, 0\}$.
Hence, after the search, one can assess how sharp the returned cutoff is by evaluating the lower certificate $\eta'$.

%
%
%

%
%
%

\section{Estimation for nonlinear properties: application to  von Neumann entropy estimation}
\label{sec:reformulation}

In this section, we apply the two-stage framework to a nonlinear property $\tr[f(\rho) \rho]$, specifically von Neumann entropy $S(\rho)=-\tr[\rho \log(\rho)]$, in which case the target function is $f(x) = \log(1/x)$. 
We particularly demonstrate an algorithm flow discussed in Section~\ref{sec:two-stage} composed of preparation, Stage~1, and Stage~2. 

First, in the preparation stage, we determine the algorithm parameters, which are the inputs to both Stage~1 and Stage~2. 
In particular, the first goal is to determine a polynomial $G$ and a bounding function $P_{\mathrm{vN}} (\delta)$ to bound the tail-induced error given in  Eq.~\eqref{eq:error_decomposition}. 
From Ref.~\cite[Lemma 11]{gilyen2019-world}, there exists an even real polynomial 
$G(x)$ of degree $D = \mathcal{O}\left(\delta^{-1} \log\left(1/\varepsilon_\mathrm{poly} \right)\right)$, which provides an approximation of a target function $f(x) = \log(x^{
-1})$ such that
\begin{equation}
\label{eq:log-polynomial}
    \left|G(x) -\frac{\log(x^{-1})}{2\log (2\pi\delta^{-1})} \right| \leq \varepsilon_\mathrm{poly}, 
\end{equation}
for all $x\in \left[\delta/\pi, 1\right]$ and $\left|G(x)\right|\le 1,~\forall x\in[-1,1]$
with $\delta, \varepsilon_\mathrm{poly} >0$. 
We determine $\alpha_F = 2\log (2\pi \delta^{-1})$ and $F(x) = \alpha_F G(x/\pi) - \log(\pi)$. Then, Eq.~\eqref{eq:error_decomposition} is modified to:
\begin{eqnarray}
& & \hspace*{-0.8em}
   \left| \tr [F(\rho) \rho ] + \tr [\log(\rho)\rho]\right|
\nonumber \\
& & \hspace*{0.8em} 
    \le 
    \alpha_F\varepsilon_{\mathrm{poly}} + (\alpha_F + \log (\pi)) C(\delta) - \tr[\rho \log(\rho) \Pi_{\rho \le \delta}]
\nonumber \\
& & \hspace*{0.8em}
    =: \alpha_F\varepsilon_{\mathrm{poly}} + \varepsilon_{\mathrm{tail}}(\delta).
\nonumber
\end{eqnarray} 
The last term of the error can be bounded by
\begin{eqnarray}
        & & \hspace*{-1em} - \tr[\rho \log(\rho) \Pi_{\rho \le \delta}] \notag \\
        & & \hspace*{-0.8em} 
          = - C(\delta)\sum_{\lambda_i \leq \delta}  \frac{\lambda_i}{C(\delta)} \left(\log\left( \frac{\lambda_i}{C(\delta)} \right) + \log\left( C(\delta) \right) \right) \notag \\
        & & \hspace*{-0.8em} 
          = - C(\delta)\sum_{i} q_i \Big( \log(q_i) + \log\left( C(\delta) \right) \Big) \notag \\
        & & \hspace*{-0.8em} \leq  C(\delta) \log\big( N/C(\delta) \big),
\end{eqnarray} 
where $q_i = \lambda_i/C(\delta)$ is a distribution satisfying $\sum_i q_i = 1$,
and we use the fact that the entropy $-\sum_i q_i \log(q_i)$ is bounded by $\log(N)$. 
Thus, the tail-induced error $\varepsilon_{\mathrm{tail}}(\delta)$ is bounded using the bounding function $P_{\mathrm{vN}}(\delta)$ as $\varepsilon_{\mathrm{tail}}(\delta) \le (\alpha_F+\log(\pi)) C(\delta) + C(\delta) \log( N/ C(\delta)) =: P_{\mathrm{vN}}(\delta)$.
    
Since $P_{\mathrm{vN}}(\delta)$ is a complicated function of $C(\delta)$, finding a $(\varepsilon, P_{\mathrm{vN}})$-spectral cutoff is nontrivial. 
Hence, below we prove that $C(\delta) \le \eta$ is a sufficient condition for $P_{\mathrm{vN}}(\delta) \le 12\eta'$ to hold, meaning that we can use $P(\delta)=C(\delta)$ rather than $P_{\mathrm{vN}}(\delta)$. 
This allows us to use Algorithm~\ref{alg-search} and
find the simplest $(\eta, C)$-spectral cutoff. 
To prove this fact of sufficiency, we begin with the observation that $\eta$ satisfying $\eta \log(N/\eta) = \eta'$ for a given $\eta' \in (0,N/e]$ is expressed by the lower branch of the Lambert $W$ function, $W_{\!-1}$, as 
    \begin{equation}
    \label{Eq W -1}
        \eta = \frac{\eta'}{-W_{\!-1}\left( -\eta'/N \right)}.
    \end{equation}
    Then, if $\delta$ satisfies $C(\delta) \le \eta$, i.e., $\delta$ is a $(\eta, C)$-spectral cutoff, we have 
    \begin{equation}
        C(\delta) \log\left(\frac{N}{ C(\delta)} \right) \le 
        \eta \log\left(\frac{N}{ \eta} \right) = \eta',
    \end{equation}
    since $x \log(N/x)$ is increasing in $x \le N/e$.
    We also assume a lower bound of $\delta$.
    Consistent with the lower bound of our algorithm output, we choose $\delta \ge \eta/(8N)$.
    Then, we have 
    \begin{align}
        \alpha_F+\log(\pi)
        &\le 
        2\log  \left(\frac{16\pi N}{\eta}\right) + \log (\pi) \notag \\
        &\le 2\eta'/\eta + 2\log (16\pi)+ \log (\pi).
    \end{align}
    Then 
    \begin{align}
        (\alpha_F+\log(\pi)) \eta &\le 2 \eta' + ( 2\log (16\pi)+ \log (\pi)) \eta \notag \\
        & \le 2 \eta' +  \frac{2\log (16\pi)+ \log (\pi)}{\log(N/\eta')} \eta'
        \le 11 \eta', \nonumber
    \end{align}
    is obtained where we used Eq.~\eqref{Eq W -1} and 
    $-W_{\!-1}\left( -\eta'/N \right) \ge \log(N/\eta')$ for $\eta' \le N /e$ in the second inequality.
    Hence, for a $(\eta, C)$-cutoff $\delta\ge \eta/(8N)$, the error bound $P_{\mathrm{vN}}(\delta) \le 12\eta'$ holds.

    We have the bounding function $P(\delta)=C(\delta)$, and we can further determine the parameters. 
    To bound the approximation error within $\varepsilon/2$, i.e.,
    $\alpha_F\varepsilon_{\mathrm{poly}} + \varepsilon_{\mathrm{tail}}(\delta) \leq \alpha_F\varepsilon_{\mathrm{poly}} + 12\eta' \le \varepsilon/2$,
    we choose $\varepsilon_{\mathrm{poly}} =\varepsilon/4\alpha_F$
    and $\eta' = \varepsilon/48$, or equivalently $\eta = \tilde{\mathcal{O}}(\varepsilon/\log(N))$.
    If $\varepsilon \le 1$, this choice ensures 
    $\eta \le N/e$ and $\eta' \le N/e$ 
    because $\eta' \le 1/48 < N/e$ and $\eta \le \eta'/\log(N/\eta') \le \eta' < N/e$.
    
    Using the parameters determined above, we now proceed to Stage~1 and Stage~2 
    to estimate von Neumann entropy. 
    In particular, the necessary number of state copies for this estimation is revealed as follows. 

\begin{lemma}[von Neumann entropy estimation]\label{lemma:vn}
    Let $\rho$ be an $N$-dimensional quantum state.
    For $\varepsilon \in (0,1]$ and $\vartheta \in (0,1)$, von Neumann entropy $S(\rho) = -\tr[\rho \log (\rho)]$ can be estimated within additive error $\varepsilon$ with a success probability at least $1-\vartheta$, using 
    $\tilde{\mathcal{O}} (\hat{\lambda}_{\eta}^{-2} \varepsilon^{-2}\log(\vartheta^{-1}))$ copies of $\rho$.
\end{lemma}

\begin{proof}
    We allocate the error budget $\varepsilon/2$ to the algorithmic bias and the statistical error.
    Also, we allocate the failure probability budget $\vartheta/2$ to each of Stage~1 and Stage~2 to ensure that the overall success probability is at least $1-\vartheta$
    by the union bound.
    In the preparation stage,
    we have determined the parameters $\eta = \tilde{\mathcal{O}}(\varepsilon)$ and $P = C$.
    In Stage~1, by Theorem~\ref{thm:search}, we find a certified $(\eta, C)$-spectral cutoff $\hat{\lambda}_{\eta}$ using $\tilde{\mathcal{O}} (\hat{\lambda}_{\eta}^{-2} \eta^{-2}\log(\vartheta^{-1})) = \tilde{\mathcal{O}} (\hat{\lambda}_{\eta}^{-2} \varepsilon^{-2}\log(\vartheta^{-1}))$ copies.
    In Stage~2, we determine the polynomial $G$ with the degree 
    $D = \mathcal{O}\!\left( \delta^{-1}\log(\varepsilon_{\mathrm{poly}}^{-1})\right)
    = \mathcal{O}\!\left(\hat{\lambda}_{\eta}^{-1}\log(\varepsilon_{\mathrm{poly}}^{-1})\right)
    =\tilde{\mathcal{O}}(\hat{\lambda}_{\eta}^{-1})$
    from the identified cutoff. 
    Now we apply Lemma~\ref{lemma:virtual-qsvt} to estimate $\tr[(\alpha_F G(\rho/\pi) - \log(\pi))\rho]$ that approximates $\tr[f(\rho)\rho]=S(\rho)$ with the target error \sout{$\mathcal{O}(\varepsilon)$} $\varepsilon/2$, where in the above we found $\alpha_F = \mathcal{O}\!\left(\log(\hat{\lambda}_\eta^{-1})\right)$. 
    As a result, we obtain the copy requirements for this estimation: $\tilde{\mathcal{O}}(D^2) = \tilde{\mathcal{O}}(\hat{\lambda}_\eta^{-2})$ copies per circuit,
    $\tilde{\mathcal{O}}(\alpha_F^2\varepsilon^{-2}\log(\vartheta^{-1}))=
    \tilde{\mathcal{O}}(\varepsilon^{-2}\log(\vartheta^{-1}))$ measurement shots,
    and consequently $\tilde{\mathcal{O}}(\hat{\lambda}_\eta^{-2}\varepsilon^{-2}\log(\vartheta^{-1}))$ total copies for Stage~2. 
    Hence, the total copy requirement through Stage~1 and Stage~2 is also
    $\tilde{\mathcal{O}}(\hat{\lambda}_\eta^{-2}\varepsilon^{-2}\log(\vartheta^{-1}))$.
\end{proof}

While our framework provides the copy requirement based on $\hat{\lambda}_\eta$, previous works mainly rely on prior knowledge of the rank $r$ of $\rho$. 
To be fairly compared to previous works, we derive the worst-case complexity based on $r$.
Since Theorem~\ref{thm:search} guarantees the lower bound $\hat{\lambda}_{\eta} \ge \max\{\lambda_{\min}/4, \eta/(8r)\}$, by directly replacing $\hat{\lambda}_{\eta}$ with the lower bounds,
we obtain the total copy requirement $\tilde{\mathcal{O}}(\hat{\lambda}_\eta^{-2} \eta^{-2}) = \tilde{\mathcal{O}}(r^2\varepsilon^{-4})$ or $\tilde{\mathcal{O}}(\lambda_{\min}^{-2}\varepsilon^{-2})$ with probability at least $2/3$. 
These are the scalings obtained in recent works~\cite{Kean2025-list, wada2025state}. 
Now, we obtain the results in Table~\ref{tab:total-complexity}. 

\section{Impact of the adaptive cutoff}
\label{sec:case-study-demonstration}

In this section, we investigate the impact of introducing the $(\eta, P)$-spectral cutoff $\lambda_\eta$ and its certified version $\hat{\lambda}_\eta$ in the case $P=C$.
Recall that they are essential to guarantee the error to estimate a target quantity, while $\lambda_{\min}$ and $\eta/r$ often provide overly loose bounds for $\lambda_\eta$. 
In particular, we demonstrate how much $\hat{\lambda}_\eta$ is bigger than the lower bounds $\lambda_{\min}$ and $\eta/r$ in specific instances of $\rho$. 
This analysis directly connects to the impact of reducing the polynomial degree for von Neumann entropy estimation. 
We first clarify this advantage through analytical examples and then demonstrate it numerically.

\subsection{Analytical examples}
\label{sec:case-study}

First, we analytically compare the three parameters: a boundary $(\eta, C)$-spectral cutoff $\lambda_\eta$, $\lambda_{\min}$, and $\eta/r$. 
Throughout the case studies in this subsection, we analyze the boundary cutoff $\lambda_\eta$ rather than the certified cutoff $\hat{\lambda}_\eta$, because $\lambda_\eta$ captures the intrinsic scaling determined by the spectrum and can be analyzed analytically.
To make the analysis consistent with a certified cutoff, we use a feasible value of $\delta$ that satisfies $C(\delta)\le \eta$ and then take the limit $\delta\to \lambda_\eta$. 

Before providing the specific instances of $\rho$, we summarize the relations among the parameters:
\begin{gather}
\lambda_{\min} \le \lambda_{\eta},~~
 \frac{\eta}{r} \le \lambda_{\eta},\label{eq:lambda-lb-ineq}\\
  \lambda_{\eta} \le \frac{1}{r_\eta}, ~~
  \lambda_{\min} \le \frac{1}{r},
\label{eq:lambda-ub-ineq}
\end{gather}
where $r_\eta$ is an $\eta$-approximate rank defined as $r_\eta = \text{Rank}(\Pi_{\lambda>\lambda_\eta}\rho)$. 
The inequalities in Eq.~\eqref{eq:lambda-lb-ineq} follow directly from their definitions, and Eq.~\eqref{eq:lambda-ub-ineq} follows from $r \lambda_{\min} \le \sum_i \lambda_i=1$ and $r_{\eta} \lambda_{\eta}\le \sum_i \lambda_i \le 1$, where $\lambda_i$ are the non-zero eigenvalues. 
These inequalities highlight the critical advantage of using $\lambda_{\eta}$ over the conventional thresholds. 
Eq.~\eqref{eq:lambda-lb-ineq} shows that the lower bounds $\lambda_{\rm min}$ and $\eta/r$ are often excessively loose, as they assume a worst-case spectral distribution.
Eq.~\eqref{eq:lambda-ub-ineq} reveals an interesting difference  between $\lambda_{\min}$ and $\lambda_{\eta}$, conceptually positioning $\lambda_{\eta}$ as the ``smoothed minimum non-zero eigenvalue'' of $\rho$.
Actually, while the exact minimum non-zero eigenvalue $\lambda_{\min}$ is bounded by the exact rank $r$, $\lambda_{\eta}$ is allowed to increase up to $1/r_{\eta}\geq 1/r$.
Note that $\lambda_{\min}$ and $r$ are highly sensitive to even small noise; thus, for low-rank states subjected to small but high-rank noise, the $\lambda_{\eta}$-input method can provide an exponential improvement over the bounds dictated by $\lambda_{\min}$ or $r$.

To illustrate the above-described advantages, we examine the following three types of $\rho$. 

\begin{table}[htbp]
    \centering
    \renewcommand{\arraystretch}{2}
    \begin{tabular}{l|ccccc}
        \hline \hline
        Instance & $\lambda_\eta$ &~~~& $\lambda_{\min}$ &~~~& $\eta/r$  \\
        \hline
        Flat spectrum  & $\dfrac{1}{r}$ &~~~& $\dfrac{1}{r}$ &~~~& $\dfrac{\eta}{r}$ \\
        Gapped spectrum & $\lambda_{\min}$ &~~~& $\lambda_{\min}$ &~~~& $\dfrac{\eta}{r}$ \\
        Noisy states & $\dfrac{1-\eta}{r_0}$ &~~~& $\dfrac{\eta}{r_1}$ &~~~& $\dfrac{\eta}{r_0 + r_1}$  \\[0.5em]
        \hline \hline
    \end{tabular}
    \renewcommand{\arraystretch}{1}
    \caption{The comparison of three parameters: $(\eta, C)$-spectral cutoff $\lambda_\eta$, $\lambda_{\min}$, and $\eta/r$ for the three examples of $\rho$. $\lambda_{\min}$ and $\eta/r$ provide lower bounds of $\lambda_\eta$.}
    \label{tab:case}
\end{table}

\textit{Flat spectrum:} a rank-$r$ quantum state satisfying
\begin{equation}
    \rho = \frac{1}{r}\sum_{i=1}^{r} \ketbra{\psi_i},
\end{equation} 
where all non-zero eigenvalues are exactly $1/r$.

\textit{Gapped spectrum:} a rank-$r$ quantum state satisfying
\begin{gather}
    \rho = \sum_{i=1}^{r} \lambda_i \ketbra{\psi_i}, \notag \\
    \lambda_1 \ge \lambda_{2} \ge \cdots \ge  \lambda_r  = \lambda_{\min} \ge \eta.
\end{gather}
This instance includes a wider class than the previous one.
For simplicity of the analysis, we assume the minimum eigenvalue is larger than the target value $\eta$.

\textit{Noisy state:} a composite state consisting of a rank-$r_0$ target signal and a rank-$r_1$ noisy component, represented as
\begin{gather}
    \rho = \frac{1-\eta}{r_0}\sum_{i=1}^{r_0} \ketbra{\psi_i} + \frac{\eta}{r_1}\sum_{i=1}^{r_1} \ketbra{\phi_i},
\end{gather}
with $(1-\eta)/r_0 \ge \eta/r_1$. In this case, $\lambda_\eta$ correctly identifies the minimum non-zero eigenvalue of the target state, $(1-\eta)/r_0$, whereas $\lambda_{\min}$ drops to the eigenvalue of the noisy state, $\eta / r_1$.

The resulting three parameters for each instance are summarized in Table~\ref{tab:case}.
In the case of the flat spectrum and the gapped spectrum, $\lambda_\eta$ improves over the worst-case bound $\eta/r$ by a factor of order $1/\eta$. 
Moreover, $\lambda_\eta$ achieves a clear advantage under noise. 
In particular, when $r_0 = \mathcal{O}(1)$ and $r_1 = \mathcal{O}(N)$, $\lambda_\eta$ presents a significant advantage over $\lambda_{\min}$ and $\eta/r$.

\subsection{Demonstration of the cutoff search}

In this subsection, we provide a numerical demonstration of Algorithm~\ref{alg-search}. 

Suppose we have a rank-$r$ quantum state $\rho = \sum_{j=0}^{r-1} \lambda_j \ketbra{j}$ that encodes a classical distribution on the computational basis.
In particular, we employ a discretized Gaussian distribution $\lambda_j \propto \exp[- (j/ \Sigma)^2]$, where $j \in [0, r-1]$ and $\Sigma >0$. 
We evaluate the estimated CDF $\hat{C}_d(x)$, which approximates the target 
$C(x) = \tr[\rho \Pi_{\rho \leq x}]$, and the resulting estimated cutoff $\hat{\lambda}_\eta$.

Fig.~\ref{fig:cdf} depicts a search process of Algorithm~\ref{alg-search} with decreasing $\delta$ from (a) $2^{-4}$ to (d) $2^{-7}$.
The parameters used for this demonstration are summarized in Table~\ref{tab:numerical-params}~(a).
In Fig.~\ref{fig:cdf}, $\hat{C}$ (red line) and $C$ (gray line) denote
the estimated and the exact CDF, respectively,
while the dashed line represents the target value $\eta = 0.1$.
Also, we observe the error bounds $\pm(\zeta+ \xi)$ (red shaded areas). 
At step (d), the algorithm successfully identifies the spectral cutoff $\hat{\lambda
}_\eta = 1.7 \times 10^{-2}$ while its upper bound is $\lambda
_\eta = 4.8 \times 10^{-2}$. 

This example demonstrates the numerical impact of the introduction of the cutoff.
Our search algorithm determines $\hat{\lambda}_\eta= 1.7 \times 10^{-2}$.
By contrast, the conventional worst-case baselines would require the eigenvalue thresholds by $\lambda_{\min} \le \exp(-(r-1)^2/\Sigma^2)\approx 10^{-18180}$
or $\eta /r = 1.0 \times 10^{-4}$.
Because the required polynomial degree scales as the inverse of the threshold, the introduction of $\hat{\lambda}_\eta$ reduces the required polynomial order by a factor of approximately $10^2$ compared to the conventional worst-case bound $\eta/r$.

\begin{figure*}[ht]
    \centering
    \begin{tabular}{cc} 
    \includegraphics[width=0.47\linewidth]{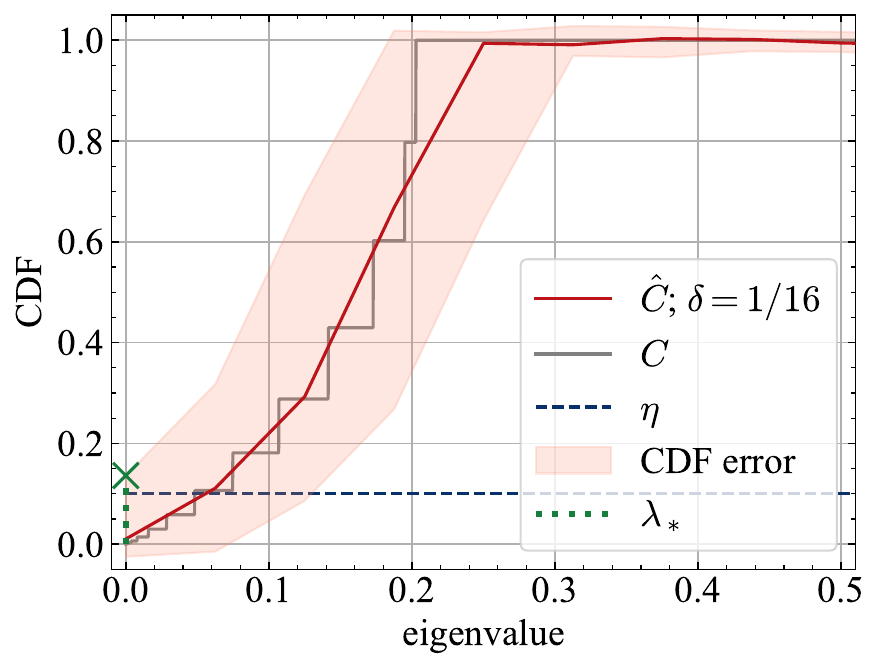}
    &\includegraphics[width=0.47\linewidth]{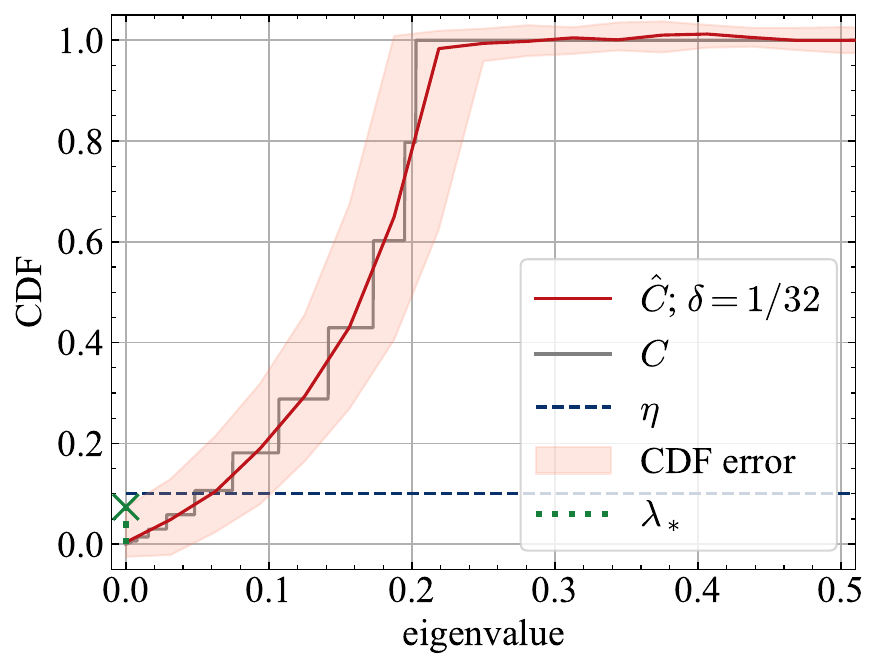}  \\
    (a) $\delta=2^{-4}$ & (b) $\delta=2^{-5}$\\
    \includegraphics[width=0.47\linewidth]{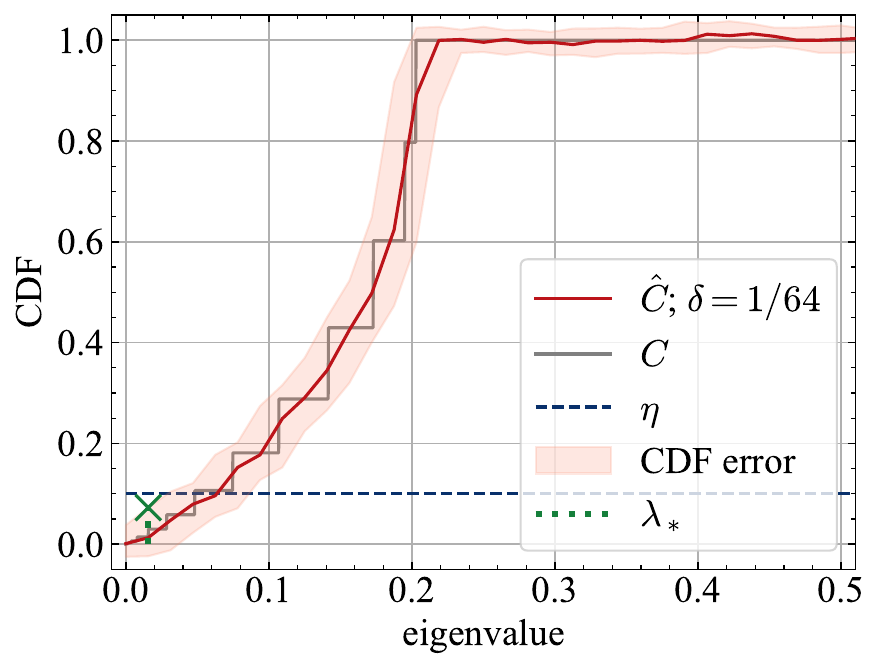}
    &\includegraphics[width=0.47\linewidth]{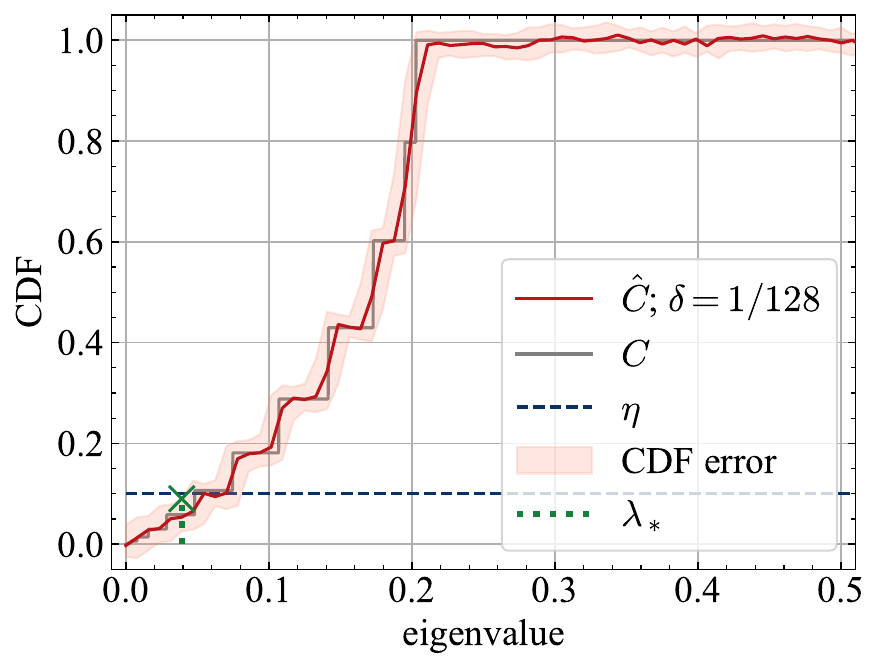}  \\
    (c) $\delta=2^{-6}$ & (d) $\delta=2^{-7}$ \\
    \end{tabular}
    \caption{The search steps to identify a certified cutoff in an instance with the boundary cutoff $\lambda_\eta = 4.8 \times 10^{-2}$.
    The flow progresses from (a) to (d) with gradually decreasing $\delta$.
    $\hat{C}$ shows the estimated CDF in red, and $C$ shows the exact CDF in gray.
    $\lambda_*$ identified by the algorithm is displayed in green.
    The red shaded areas display the CDF approximation error bound implied by Eq.~\eqref{eq:cdf-approximation-bound}.
    (a, b) There are no grid points satisfying Eq.~\eqref{eq:find-lambda-star}, and $\lambda_* = 0$ is returned. 
    (c) A non-zero $\lambda_*$ is found, but the algorithm continues since $\lambda_* < 2\delta_l$.
    (d) $\lambda_* \ge 2\delta$ is found, and the algorithm stops with the return $\hat{\lambda}_\eta = 1.7 \times 10^{-2}$.
    }
    \label{fig:cdf}
\end{figure*}

\begin{table}[ht]
    \centering
    \begin{tabular}{ccc}
        \begin{tabular}{c|c}\hline \hline
            parameters &  values\\ \hline
            $r$ &  1024 \\ 
            $\lambda_j$ & $\propto \exp[-j^2/\Sigma^2]$  \\
            $\Sigma$ & $5$  \\
            $\eta$ & $0.1$  \\
            \hline \hline
        \end{tabular}
        & ~~
        &
        \begin{tabular}{c|c}\hline \hline
            parameters &  cutoff values\\ \hline
            $\lambda_\eta$ & $4.8 \times 10^{-2}$  \\
            $\hat{\lambda}_\eta$ & $1.7 \times 10^{-2}$  \\
            $\lambda_{\min}$ & $\approx10^{-18180}$ \\
            $\eta /r$ & $1.0 \times 10^{-4}$ \\
            \hline \hline
        \end{tabular}
        \\ 
        (a) The parameter setup & ~~ &
        (b) The threshold result
    \end{tabular}
        \caption{Parameters in the numerical analysis.}
    \label{tab:numerical-params}
\end{table}


\section{Conclusion}\label{sec:conclusion}
In this work, we propose a method for reducing the prohibitive cost of QSVT-based property estimation by removing the use of the worst-case assumptions characterized by the minimum non-zero eigenvalue or the rank of density matrices. 
By introducing the two-stage algorithm based on the spectral cutoff, our algorithm can significantly reduce the polynomial degree of QSVT and accordingly improve the total number of state copies. 
To realize the idea of the spectral cutoff, we develop the search algorithm for the spectral cutoff, which returns the cutoff with a certified interval with high probability. 

For the application of the framework, we mainly focused on estimating von Neumann entropy.
In the task, our copy requirement matches the best-known complexity~\cite{wada2025state, Kean2025-list} without prior knowledge of $\lambda_{\min}$ or $r$, meaning that our technique is applicable to completely unknown quantum states or high-rank states, without any drawbacks in complexity. 

We leave several open questions for future investigation.
One promising direction is to expand the range of the target tasks. 
While we demonstrate two primary tasks, estimating other important properties, such as the R{\'e}nyi entropy with $\alpha < 1$ and the quantum relative entropy, remains a challenge. 
To deal with a new task within our framework, we need to derive a pair $(\eta, P)$ such that the certified $(\eta, P)$-spectral cutoff can be efficiently identified by the search algorithm.
Because our current search algorithm is limited to specific types of $P$, identifying solvable and sufficiently tight conditions for broader classes of properties is a crucial next step.
Furthermore, Algorithm~\ref{alg-search} can also be naturally extended to the purified-query access model. Indeed, the purified query yields the block-encoding of $\rho$, and the block-encoding generates the Hamiltonian simulation, which is the input of Algorithm~\ref{alg-search}.
Establishing this extension and deriving the corresponding complexities is an important direction for future work. 

Furthermore, our two-stage framework has the potential to be applied to the broader context beyond the estimation of nonlinear properties of quantum states. 
The search algorithm (Algorithm~\ref{alg-search}) can be extended to analyze the spectra of a wide class of Hamiltonians. 
For instance, quantum linear system solvers often suffer from a loose estimate of the minimum non-zero eigenvalues.
Our cutoff might improve their complexities via an effective state-dependent preconditioning.
Exploring how the two-stage framework improves the practical complexity of those problems is an important avenue for future research.

\begin{acknowledgments}
This work was supported by MEXT Quantum Leap Flagship Program Grants No. JPMXS0118067285 and No. JPMXS0120319794. J.K. was supported by SIP Grant No. JPJ012367. A.T. was supported by JST SPRING, Japan Grant No. JPMJSP2123. H.H. was supported by JSPS KAKENHI Grant No. JP26KJ1951. K.W. was supported by JSPS KAKENHI Grant No. JP24KJ1963 and JST ASPIRE Grant No. JPMJAP2316.
\end{acknowledgments}

\bibliography{bibtex}
\clearpage
\onecolumngrid
\setcounter{equation}{0}
\setcounter{section}{0}
\appendix

\section{Fourier series approximation of the Heaviside step function}
\label{appendix:fourier-heaviside}

We review the Fourier series approximation of the Heaviside step function $\Theta(x)$ proposed by Wan et al.~\cite{wan2022randomized}. 
The Fourier coefficients $F_{j,d}$ are explicitly given in the Supplementary Material of Ref.~\cite{wan2022randomized} as follows: 
\begin{gather}\label{eq:wan_A2}
    F(x) = \sum_{j=-d}^{d} F_{j, d} e^{ijx}, \notag\\
    F_{0,d} = \frac{1}{2}, ~~
    F_{j,d} = 0 \text{ for even } j \text{  and }  j > 0 , \notag \\
    F_{j, d} = - i \sqrt{\frac{\beta}{2 \pi}} e^{-\beta} \frac{I_{(j-1)/2} (\beta) + I_{(j+1)/2}(\beta)}{j} \text{ for odd }  j \text{ and }  0< j < d , \notag\\
    F_{d,d} = - i \sqrt{\frac{\beta}{2 \pi}} e^{-\beta} \frac{I_{(d-1)/2} (\beta)}{d},\notag\\
    F_{j, d} = -F_{-j,d}\text{ for }  -d \le j < 0,
\end{gather}
where $d$ is a positive odd integer representing the degree of the approximation,
$\beta > 0$ is a free parameter, and $I_\nu(z)$ denotes the modified Bessel function of the first kind of order $\nu$.
$F(x)$ provides an approximation of the Heaviside step function $\Theta(x)$ for $\delta \in (0,\pi/2)$ and $\zeta > 0$ such that
\begin{gather}
    |\Theta(x) - F(x)| \leq \zeta, ~~~ \forall x\in [-\pi + \delta, -\delta] \cup [\delta, \pi-  \delta], \notag\\
    - \zeta \le F(x) \le 1 + \zeta, ~~~ \forall x\in[-\pi, \pi] \label{eq:heaviside-bound},
\end{gather}
if the parameters are determined as follows:
\begin{equation}\label{eq:def-d}
    d = g(\delta, \zeta) 
    = \begin{cases}
        2\lceil \sqrt{ w_{\zeta} f(\beta,\zeta_*) } \rceil +1& \zeta_* < 1 \\
        2\lceil \sqrt{ w_{\zeta} \beta } \rceil +1& \zeta_* \ge 1,
    \end{cases}
\end{equation}
where
\begin{gather}
    f(\beta,\zeta_*) = \frac{\log(\zeta_*^{-1})- \beta}{W_0\left( \frac{1}{e \beta} (\log(\zeta_*^{-1})- \beta) \right)},
    ~~
    w_{\zeta} = W_0\left(\frac{18}{\pi \zeta^2}\right),~~
   \zeta_* = \frac{2\sqrt{2 \pi w_\zeta} \zeta}{3}, \notag \\
    \beta = \beta(\delta, \zeta) = \max \left\{ \frac{1}{4 \sin^2(\delta)} W_0\left(\frac{3 }{\pi \zeta^2}\right), 1 \right\},
    \label{eq:def-f-beta}
\end{gather}
and $W_0(x)$ is the principal branch of the Lambert $W$ function.
The function $f(\beta, \zeta_*)$ has the properties that
$f(\beta, \zeta_*) > \beta$ for $\zeta_* < 1$ and 
$f(\beta, \zeta_*) = \mathcal{O}(\beta + \log(1/\zeta_*))$ from Proposition 9 in Ref.~\cite{wan2022randomized}.
From the behavior of $f(\beta, \zeta_*)$, the degree $d$ scales as $d = \mathcal{O}(\delta^{-1} \log(1/\zeta))$.

\section{Application: R{\'e}nyi entropy ($\alpha>1$)}\label{appendix:renyi}

In this section, we apply the proposed two-stage framework to R{\'e}nyi entropy $R_\alpha(\rho ) = \frac{1}{1-\alpha}\log(\tr [\rho^\alpha])$ estimation, in the case $\alpha > 1$. 
R{\'e}nyi entropy estimation also presents a challenge that differs from the von Neumann entropy estimation.
Conventionally, a quantum computer is used to estimate the R{\'e}nyi moment $M_\alpha (\rho) = \Tr[\rho^{\alpha}]$, and then $R_\alpha(\rho)$ is classically computed by taking the logarithm. 
Note that to estimate $R_\alpha(\rho)$ in additive error $\varepsilon$, we should estimate $M_\alpha (\rho)$ in multiplicative error. 
That often requires lower and upper bounds of $\Tr[\rho^{\alpha}]$, while the worst-case bound is highly conservative.
In our two-stage framework, we obtain the lower bound of $\Tr[\rho^{\alpha}]$ in the identification stage in addition to the spectral cutoff. 
We leave the case of $(0< \alpha < 1)$ for future work.

\subsection{Preparation}
First, in the preparation stage, we determine a bounding function $P_\alpha (\delta)$ and a target value of the error tolerance $\eta$.
We construct the polynomial $F(x)$ as a product of the step function and $x^{\alpha-1}$, similar to the method in Ref.~\cite{wang2024discrete}.
    From Ref.~\cite[Corollary 15]{gilyen2022-fidelity}, there exists an even real polynomial $H(x)$ of degree
    $\mathcal{O}\left(\delta^{-1} \log\left(1/\varepsilon_\mathrm{poly} \right)\right)$,
    which provides an approximation of the step function in positive $x$ such that
    \begin{gather}\label{eq:step-polynomial}
        H(x) \in [0, \varepsilon_\mathrm{poly}] ~~~\forall x \in {[0 , \delta/2]}, \notag \\
        H(x) \in [1-\varepsilon_\mathrm{poly}, 1] ~~~\forall x \in {[\delta, 1]}, \notag \\
        \left| H(x) \right| \le 1~~~\forall x\in[-1,1],
    \end{gather}
    where $\delta, \varepsilon_\mathrm{poly} >0$. 
    Similarly, from Ref.~\cite[Lemma 17]{gilyen2022-fidelity}, 
    there exists an even real polynomial $J(x)$ of degree
    $\mathcal{O}\left(\delta^{-1} \log\left(1/\varepsilon_\mathrm{poly} \right)\right)$,
    which provides an approximation of the function $x^{c}/ 2$ for $c > 0$ such that
    \begin{gather}
        \left|J(x) - \frac{1}{2}x^{c} \right| \leq \varepsilon_\mathrm{poly}~~~\forall x\in {[\delta/2, 1]}, \notag \\
        \left|J(x)\right|\le 1~~~\forall x\in[-1,1].
    \end{gather}
    By setting $c = \alpha - 1 > 0$, we have the even real polynomial $G(x) = H(x) J(x)$ of degree
    $D = \mathcal{O}\left(\delta^{-1} \log\left(1/\varepsilon_\mathrm{poly} \right)\right)$ such that
    \begin{gather}
        |G(x)| \le \varepsilon_\mathrm{poly} ~~~\forall x \in {[0 , \delta/(2\pi)]}, \notag \\
        |G(x)| \le \frac{1}{2}x^{\alpha-1} + \varepsilon_\mathrm{poly} ~~~\forall x \in {[\delta/(2\pi), \delta/\pi]}, \notag \\
        \left|G(x) - \frac{1}{2}x^{\alpha-1}\right| \le 2\varepsilon_\mathrm{poly} ~~~\forall x \in {[\delta/\pi, 1]}, \notag \\
        \left|G(x) \right| \le 1~~~\forall x\in[-1,1].
    \end{gather}
    We define $F(x) = \alpha_F G(x /\pi)$. This yields the approximation error bound as
    \begin{align}\label{eq:renyi-moment-error}
        \left| \tr[F(\rho) \rho ] - \tr[\rho^\alpha] \right|
        &\leq 2 \alpha_F \varepsilon_\mathrm{poly} + \delta^{\alpha - 1} C(\delta) + \tr[\rho^{\alpha}  \Pi_{\rho\le\delta}] \notag \\
        &\leq 2 \alpha_F\varepsilon_\mathrm{poly} + 2\delta^{\alpha - 1} C(\delta),
    \end{align}
    where $\alpha_F  = 2\pi^{\alpha - 1}$ and $\Pi_{\rho \le \delta}$ is a projector onto the subspace with eigenvalues $\lambda \leq \delta$.
    At the last inequality, we use the fact that $\tr[\rho^{\alpha}  \Pi_{\rho \le \delta}] \leq \delta^{\alpha -1} C(\delta)$.
    To maintain the error in Eq.~\eqref{eq:renyi-moment-error} within $\varepsilon'$, we should choose the parameters $\varepsilon_\mathrm{poly} = \varepsilon'/(4\alpha_F)$, and $\delta$ satisfying
    \begin{equation}\label{eq:renyi-p-tail}
        P_{\mathrm{tail}}(\delta) = 2\delta^{\alpha-1} C(\delta) \le \frac{\varepsilon'}{2}.
    \end{equation}
    
    Next, we confirm the error propagation to the overall entropy estimation.
    That is, we determine $\varepsilon'$ for a given target accuracy $\varepsilon$.
    Once we obtain the estimation of $\mathrm{tr}[F(\rho)\rho ]$ within additive error  $\varepsilon'(\leq\tr[\rho^\alpha]/2)$, the propagated error is bounded as
    \begin{align}
        & \left| \frac{1}{1-\alpha} \log (\tr[\rho^\alpha]) - \frac{1}{1-\alpha} \log (\tr [F(\rho) \rho ]) \right| \notag \\
        & ~~~ \le \frac{1}{\alpha - 1} \left| \log \left( 1 - \frac{\varepsilon'}{\tr[\rho^\alpha]} \right) \right| 
        \leq  \frac{2\varepsilon'}{(\alpha - 1)\tr[\rho^\alpha]},
    \end{align}
    where we use $-\log(1-x) \le 2x$ for $x \in [0, 1/2]$ at the last inequality.
    
    To bound the denominator, we need a lower bound of $\tr[\rho^\alpha]$.
    We observe that
    $\tr[\rho^\alpha] \ge \lambda^{\alpha-1}\tr[\rho (\bm{1}-\Pi_{\lambda}) ] = \lambda^{\alpha-1} (1-C(\lambda))$, for any eigenvalue $\lambda \in [0, 1]$.
    Particularly, when we set $\lambda$ as a certified $(1/2, C)$-spectral cutoff $\hat{\lambda}_{1/2, C}$ such that $ C(\hat{\lambda}_{1/2, C}) \le 1/2$,
    $\tr[\rho^\alpha] \geq \hat{\lambda}_{1/2, C}^{\alpha -1}/2$ is obtained.
    Thus, the bias of the entropy $\frac{1}{1-\alpha} \log(\tr[\rho^\alpha])$ is bounded by 
    additive error $\frac{4\varepsilon'}{\hat{\lambda}^{\alpha-1}_{1/2, C}(\alpha-1)}$.

    Thus, we choose $\varepsilon'$ as 
    \begin{equation}\label{eq:def-varepsilon-prime}
            \varepsilon'
            = \frac{1}{4} \hat{\lambda}^{\alpha-1}_{1/2, C}\min\{(\alpha-1)\varepsilon,1\}
            = \frac{1}{4} \hat{\lambda}^{\alpha-1}_{1/2, C}(\alpha-1)\varepsilon,
    \end{equation}
    where we assume $\varepsilon \le \frac{1}{\alpha-1}$.
    Therefore, with Eq.~\eqref{eq:renyi-p-tail}, we choose $P_\alpha$ and $\eta$ for this task as
    \begin{equation}\label{eq:def-p-alpha}
        P_\alpha(\delta) = \left(\frac{\delta}{\hat{\lambda}_{1/2, C}}\right)^{\alpha -1} C(\delta) \le \frac{\alpha-1}{16}\varepsilon
         = \eta.
    \end{equation}
    Using the $(\eta, P_\alpha)$-spectral cutoff $\hat{\lambda}_\eta$, we can bound the tail-induced error of the R{\'e}nyi entropy within $\varepsilon$.
    The polynomial degree $D$ is obtained by $D = \mathcal{O}\left(\delta^{-1} \log\left(1/\varepsilon_\mathrm{poly} \right)\right) 
    =\mathcal{O}\left(\hat{\lambda}_{\eta}^{-1} \log\left(\hat{\lambda}_{1/2, C}^{-1}\varepsilon^{-1}\right)\right)=\tilde{\mathcal{O}}(\hat{\lambda}_{\eta}^{-1})$.
    
\subsection{Identification}
Here we introduce a search algorithm for the $(\eta, P_\alpha)$-spectral cutoff.

\begin{lemma}[Search algorithm for $(\eta, P_\alpha)$-spectral cutoff]\label{lemma:renyi-search}
    Let $\rho$ be an unknown but accessible quantum state
    and consider the spectral cutoff search with $H=\rho$.
    Let $\eta \in (0,1/2]$ be a target value and $\vartheta \in (0,1)$ be a failure probability. 
    Let $\hat{\lambda}_{1/2, C}$ denote a certified
    $(1/2,C)$-spectral cutoff identified by
    Algorithm~\ref{alg-search} with target value $1/2$.
    For $\alpha>1$, define a bounding function $P_\alpha(x)$ as 
    \begin{equation}
        P_\alpha(x)=\left(\frac{x}{\hat{\lambda}_{1/2, C}}\right)^{\alpha-1} C(x).
    \end{equation}
    Then, there exists a quantum algorithm that outputs a certified $(\eta, P_\alpha)$-spectral cutoff $\hat{\lambda}_\eta$ with a success probability at least $1-\vartheta$,
    such that 
    $P_\alpha(\hat{\lambda}_\eta)  \le \eta$ and 
    $\hat{\lambda}_\eta \le \hat{\lambda}_{1/2,C}$.
    This algorithm requires $\tilde{\mathcal{O}}\left(\hat{\lambda}_{\eta}^{-2} \eta^{-2}\log(\vartheta^{-1})\right)$ total copies.
    Moreover, $\hat{\lambda}_\eta$ satisfies
    \begin{equation}
        \hat{\lambda}_\eta = \Omega(\eta^{1/\alpha}/r),
    \end{equation}
    where $r$ is the rank of $\rho$.
\end{lemma}

\begin{proof}
First, we run Algorithm~\ref{alg-search} with target value $1/2$ and obtain a certified $(1/2, C)$-spectral cutoff $\hat{\lambda}_{1/2, C}$.
Then, we find a $(\eta, P_\alpha)$-spectral cutoff $\hat{\lambda}_{\eta}$.
We allocate failure probability $\vartheta/2$ to the search for $\hat{\lambda}_{1/2, C}$ and $\vartheta/2$ for $\hat{\lambda}_\eta$.
The first search requires $\tilde{\mathcal{O}}(\hat{\lambda}_{1/2, C}^{-2}\log(\vartheta^{-1}))$ copies by Theorem~\ref{thm:search}.

Then, using the determined $\hat{\lambda}_{1/2, C}$, 
we define an estimator 
$\hat{P}^{(v)}_d(x) = (x/\hat{\lambda}_{1/2, C})^{\alpha-1} \hat{C}_d^{(v)}(x)$, which can be constructed from $\hat{C}_d^{(v)}(x)$ by classical post-processing.
$\hat{C}_d^{(v)}(x)$ is defined in Theorem~\ref{thm:search}.
We restrict classical search to the interval $[0, \hat{\lambda}_{1/2, C}-\delta_{l}]$
so that the prefactor is bounded by one
on the region. For a fixed iteration $l$, we obtain that
\begin{gather}
    \hat{P}_{d_l}^{(v)}(x)  \geq \left(\frac{x}{\hat{\lambda}_{1/2, C}}\right)^{\alpha-1} \big(C(x - \delta_l)  - E \big) \geq P_\alpha(x - \delta_l)   - E, \notag \\
    \hat{P}_{d_l}^{(v)}(x) \leq \left(\frac{x}{\hat{\lambda}_{1/2, C}}\right)^{\alpha-1} \big(C(x + \delta_l)  +E \big) \leq P_\alpha(x + \delta_l)  +E,\label{eq:pvdme-bound}
\end{gather}
holds with probability at least $1 -\vartheta_l / 2$
for all grid points $x \in  \chi_{\delta_l}$,
when $E = \zeta + \xi + \kappa$ and 
$N_{d_l} = \left\lceil\frac{4e^4\mathcal{F}_{d_l}^2}{\xi^2}  \log(4|\chi_{\delta_l}|\vartheta_{l}^{-1}) \right\rceil$.
$\vartheta_l$ is allocated as in Theorem~\ref{thm:search}.
By replacing $\hat{C}_{d_l}(x)$ with $\hat{P}^{(v)}_{d_l}(x)$,
and choosing $\zeta=\xi=\kappa =\eta/ 12$
and $N_{d_l} = \left\lceil\frac{4e^4\mathcal{F}_{d_l}^2}{\xi^2}  \log(4|\chi_{\delta_l}|\vartheta^{-1}_l) \right\rceil$
in Algorithm~\ref{alg-search},
we obtain $\hat{\lambda}_{\eta}=2\delta_l$
satisfying 
$P_\alpha (\hat{\lambda}_\eta) \le \hat{P}^{(v)}_{d_l}(\hat{\lambda}_\eta + \delta_l) +E \le \eta$.
This algorithm has the same scaling of
the copy requirement as Theorem~\ref{thm:search}, 
$\tilde{\mathcal{O}}\!\left(\hat{\lambda}_{\eta}^{-2} \eta^{-2}\log(\vartheta^{-1})\right)$.
Since $\hat{\lambda}_\eta\le \hat{\lambda}_{1/2, C}$ from the search region, and $\eta\le 1/2$, 
the total copy requirement is dominated by the second search,
i.e., 
$\tilde{\mathcal{O}}\!\left(\hat{\lambda}_{\eta}^{-2} \eta^{-2}\log(\vartheta^{-1})\right)+\tilde{\mathcal{O}}(\hat{\lambda}_{1/2, C}^{-2}\log(\vartheta^{-1})) = \tilde{\mathcal{O}}\!\left(\hat{\lambda}_{\eta}^{-2} \eta^{-2}\log(\vartheta^{-1})\right)$.

Next, we confirm the termination of the algorithm. 
Consider a test point $x_{\text{test}} = 2\delta_{\text{test}}$.
Since $x_{\text{test}}$ must be included in the search region $[0, \hat{\lambda}_{1/2,C}]$,
we require $x_{\text{test}} \le 1/(24r)$
to satisfy $x_{\text{test}} + \delta_{\text{test}} \le 1/(16r)\le \hat{\lambda}_{1/2,C}$, from Theorem~\ref{thm:search}.
From Eq.~\eqref{eq:pvdme-bound},
we have
\begin{align}
    \hat{P}^{(v)}_{d_l}(x_{\mathrm{test}}+\delta_{\mathrm{test}}) &\leq \left(\frac{x_{\mathrm{test}}+\delta_{\mathrm{test}}}{\lambda_{1/4,C}/4}\right)^{\alpha-1} C(x_{\mathrm{test}}+2\delta_{\mathrm{test}})  +E\notag \\
    &\le \left(24 x_{\mathrm{test}}r \right)^{\alpha-1} \left(2 x_{\mathrm{test}} r\right)  +E \notag \\
    &=\frac{1}{12}\left(24 x_{\mathrm{test}}r\right)^{\alpha}   +E,
\end{align}
where we use $C(x) \le rx$ at the second inequality.
If $x_{\text{test}} \le \eta^{1/\alpha}/(24r) \le 1/(24r)$,
the termination condition $\hat{P}^{(v)}_{d_l}(x_{\mathrm{test}}+\delta_{\mathrm{test}}) \le \eta/12 + E \le \eta - E$ holds with $E =\eta /4$.
There exists a dyadic resolution $\delta_{\text{ test}}= 2^{-l} $ for a positive integer $l$ with $\eta^{1/\alpha}/(96r) \le \delta_{\text{test}}<\eta^{1/\alpha}/(48r)$.

Hence, the algorithm terminates at iteration $l$ or earlier.
Then, we obtain $\hat{\lambda}_\eta \ge 2 \delta_{\mathrm{test}} = \eta^{1/\alpha}/(48r)$.
\end{proof}

\subsection{Estimation and total copy requirement}
We finally obtain the complexity for estimating R{\'e}nyi entropy.
\begin{lemma}[R{\'e}nyi entropy $(\alpha > 1)$ estimation]\label{lemma:Renyi-poly}
    Let $\rho$ be an $N$-dimensional quantum state.
    For a fixed $\alpha>1$, $\varepsilon \in (0, 1/(\alpha-1)]$ and $\vartheta \in (0,1)$,
    R{\'e}nyi entropy $\frac{1}{1-\alpha}\log (\mathrm{tr}[\rho^\alpha]) $ can be estimated within additive error $\varepsilon$ with a success probability at least $1-\vartheta$,
    using 
    $\tilde{\mathcal{O}} (\hat{\lambda}_{\eta}^{-2} \hat{\lambda}_{1/2, C}^{2-2\alpha}\varepsilon^{-2} \log(\vartheta^{-1}))$ copies of $\rho$,
    where $\hat{\lambda}_\eta$ is a certified $(\eta, P_{\alpha})$-spectral cutoff defined by
\begin{equation}\label{eq:renyi-p-def}
    P_{\alpha}(\hat{\lambda}_\eta) = \left(\frac{\hat{\lambda}_\eta}{\hat{\lambda}_{1/2, C}}\right)^{\alpha-1}C(\hat{\lambda}_\eta) \le \eta.
\end{equation}
The target value $\eta$ is given by $\eta = \Omega(\varepsilon)$,
and $\hat{\lambda}_{1/2,C}$ denotes a certified $(1/2, C)$-spectral cutoff.
\end{lemma}

\begin{proof}
    We allocate failure probability $\vartheta/2$ to Stage~1 and $\vartheta/2$ to Stage~2.
    By the union bound, the overall success probability is at least $1-\vartheta$.
    We allocate half of the error budget for the moment estimation $\varepsilon'$ defined in Eq.~\eqref{eq:def-varepsilon-prime} to the approximation bias of $\tr[F(\rho)\rho]$, and the other half to the statistical error in Stage~2.
    In the preparation stage, we have determined the parameters $\eta = \frac{\alpha-1}{16} \frac{\varepsilon}{2}$ and $P_\alpha$.
    In Stage~1, by Lemma~\ref{lemma:renyi-search}, we find a certified $(\eta, P_\alpha)$-spectral cutoff $\hat{\lambda}_{\eta}$ using $\tilde{\mathcal{O}} (\hat{\lambda}_{\eta}^{-2} \eta^{-2} \log (\vartheta^{-1})) = \tilde{\mathcal{O}} (\hat{\lambda}_{\eta}^{-2} \varepsilon^{-2}\log (\vartheta^{-1}))$ copies.
    Here, the copy requirement in Lemma~\ref{lemma:renyi-search} includes both the search for $\hat{\lambda}_{1/2,C}$ and the subsequent search for $\hat{\lambda}_\eta$.
    In Stage~2, we determine the polynomial $F(x) = \alpha_F G(x/\pi)$ with the degree 
    $D = \tilde{\mathcal{O}}(\hat{\lambda}_{\eta}^{-1})$ from the identified cutoff, where in the above we found $\alpha_F = 2 \pi^{\alpha-1}$.
    Now we apply Lemma~\ref{lemma:virtual-qsvt} to estimate $\tr[ F(\rho) \rho]$ that approximates $\tr[\rho^\alpha]$ with the target error $\varepsilon'/2$.
    As a result, we obtain the copy requirements for this estimation: $\tilde{\mathcal{O}}(D^2) = \tilde{\mathcal{O}}(\hat{\lambda}_\eta^{-2})$ copies per circuit,
    $\tilde{\mathcal{O}}(\alpha_F^2 \varepsilon'^{-2}\log(\vartheta^{-1}))=
    \tilde{\mathcal{O}}(\varepsilon^{-2} \hat{\lambda}_{1/2,C}^{-2\alpha+2}\log(\vartheta^{-1}))$ measurement shots,
    and consequently $\tilde{\mathcal{O}}(\hat{\lambda}_\eta^{-2}\hat{\lambda}_{1/2,C}^{-2\alpha+2}\varepsilon^{-2}\log(\vartheta^{-1}))$ copies for Stage~2.
    Since $\alpha>1$ and $\hat{\lambda}_{1/2,C}\le 1$, the copy requirement of Stage~2 dominates that of Stage~1.
    Hence, the total copy requirement of Stage~1 and Stage~2 is  $\tilde{\mathcal{O}}(\hat{\lambda}_\eta^{-2}\hat{\lambda}_{1/2,C}^{-2\alpha+2}\varepsilon^{-2}\log(\vartheta^{-1}))$.
\end{proof}

To compare the copy requirements with previous works, we show the rank-based worst-case bound.
We observe that $\hat{\lambda}_\eta =\Omega(\varepsilon^{1/\alpha}/r)$ by Lemma~\ref{lemma:renyi-search} and $\eta\sim\varepsilon$. Also, $\hat{\lambda}_{1/2, C} = \Omega(1/r)$ by Theorem~\ref{thm:search}.
Then, the estimation requires  $\tilde{\mathcal{O}}(\varepsilon^{-2-2/\alpha} r^{2\alpha})$
copies of $\rho$ in the worst case.
The copy requirements are summarized in Table~\ref{tab:renyi-complexity}.

\begin{table}[htbp]
    \centering
    \renewcommand{\arraystretch}{1.5}

    \begin{tabular}{ccc|c}\hline \hline
        Identification  & Estimation & Total & ~~~ Total (worst case)\\
        \hline
        $\tilde{\mathcal{O}}\!\left(\hat{\lambda}_{\eta}^{-2} \varepsilon^{-2}\right)$ ~~~ & $\tilde{\mathcal{O}}\!\left(\hat{\lambda}_{\eta}^{-2} \hat{\lambda}_{1/2,C}^{-2\alpha+2}\varepsilon^{-2}\right)$ ~~~ &
        $\tilde{\mathcal{O}}\!\left(\hat{\lambda}_\eta^{-2}\hat{\lambda}_{1/2,C}^{-2\alpha+2}\varepsilon^{-2}\right)$ ~~~ &
        ~~~ $\tilde{\mathcal{O}}(\varepsilon^{-2-2/\alpha} r^{2\alpha})$\\
        \hline \hline
    \end{tabular}
    \renewcommand{\arraystretch}{1}
    \caption{
    The necessary number of copies of $\rho$ for estimating R{\'e}nyi entropy for a fixed $\alpha$ within an additive error $\varepsilon$ with probability at least $2/3$; 
    the copies are consumed in the identification stage (Stage 1) and the estimation stage (Stage 2). 
    $\hat{\lambda}_{\eta}$ is a certified cutoff satisfying $P_\alpha(\hat{\lambda}_{\eta}) \le \eta$ with $\eta$ the target error tolerance value, and $\hat{\lambda}_{1/2, C}$ is a certified $(1/2, C)$-spectral cutoff.
    The worst-case copy requirements are described in terms of the rank $r$.
    $\tilde{\mathcal{O}}(\cdot)$ hides logarithmic factors of $\varepsilon$, $N$, $\hat{\lambda}_{\eta}$, $\hat{\lambda}_{1/2, C}$.}
    \label{tab:renyi-complexity}
\end{table}

For quantum R{\'e}nyi entropy in the sample access model, several algorithms have been proposed. 
Recent work~\cite{Kean2025-list} derives an upper bound of 
$\tilde{\mathcal{O}}(\varepsilon^{-2-2/\alpha}r^{2})$ via the sample-to-query lifting technique.
This scaling is better than ours ($r^{2}$ vs $r^{2\alpha}$). 
They applied the lifting theorem to Ref.~\cite{wang2024discrete}, which estimates R{\'e}nyi entropy in the purified-query access model.
Compared with Ref.~\cite{wang2024discrete},
the polynomial degree $D$ of our algorithm recovers the scaling of Ref.~\cite{wang2024discrete}.
The cost difference comes from the estimation stage.
They introduce a sophisticated estimation protocol involving variable-time amplitude estimation. 
In contrast, we employ the simple estimation protocol via the virtual DME. 
However, we note that our search algorithm is independent of the following estimation protocol used, meaning that it can be combined with other estimation protocols.
Such a combination may result in further cost reductions.
\end{document}